%% file: main.tex
\documentclass[11pt]{article}

\usepackage[T1]{fontenc}
\usepackage[utf8]{inputenc}
\usepackage[a4paper,margin=27mm]{geometry}
\usepackage{amsmath,amssymb,amsfonts,mathtools}
\usepackage{bm}
\usepackage{booktabs}
\usepackage{graphicx}
\usepackage{longtable}
\usepackage{enumitem}
\usepackage{nicefrac}
\usepackage{url}
\usepackage{xcolor}
\usepackage{float}
\usepackage{caption}
\usepackage{microtype}
\usepackage[numbers,sort&compress]{natbib}
\usepackage[hidelinks]{hyperref}
\usepackage{amsthm}

\newcommand{\mathbbm}[1]{\mathbb{#1}}
\newtheorem{theorem}{Theorem}
\newtheorem{lemma}{Lemma}

\newtheorem{corollary}{Corollary}
\theoremstyle{remark}
\newtheorem{remark}{Remark}

\graphicspath{{./}{basic_interval_outputs_v5/}{generated_figures_main/}}

\DeclareMathOperator{\Var}{Var}
\DeclareMathOperator{\Cov}{Cov}
\DeclareMathOperator{\CI}{CI}
\DeclareMathOperator{\E}{\mathbb{E}}

\title{Wild Bootstrap and Efron's Bootstrap for Debiased Cox Regression}
\author{%
Lena Schemet$^{1}$ \qquad Sarah Friedrich-Welz$^{1,2}$\\[0.6em]
\small $^{1}$Mathematical Statistics and Artificial Intelligence in Medicine, University of Augsburg, Augsburg, Germany\\
\small $^{2}$Center for Advanced Analytics and Predictive Sciences (CAAPS), University of Augsburg, Augsburg, Germany\\[0.4em]
\small Correspondence: lena.schemet@uni-a.de
}
\date{August 2026}

\begin{document}
\maketitle

\begin{abstract}
Cox regression with Lasso penalization is widely used for variable selection in time-to-event data, but reliable coefficient inference after selection remains difficult. We investigate bootstrap inference for the debiased Cox estimator after Cox Lasso selection. Two score-based procedures are considered: a wild bootstrap using independent multipliers and an Efron bootstrap using centered multinomial resampling weights. Both procedures keep the original Cox Lasso fit and debiasing matrix fixed and apply bootstrap weights to the score contributions that determine the first-order distribution of the debiased estimator. We establish asymptotic validity for both bootstrap schemes and study their finite-sample behavior in extensive simulations. The bootstrap intervals improve coverage accuracy over first-order debiased Wald intervals in many small- and moderate-sample settings, with the size of the improvement depending on the simulation setting and tuning rule. A SUPPORT2 application illustrates the procedures in practice.
\end{abstract}


\section{Introduction}
\label{sec:intro}

The Cox proportional hazards model \cite{cox1972regression}
is a cornerstone of survival analysis and is
widely used in biomedical and epidemiological research
\cite{therneau2000cox}.
In classical low-dimensional settings, inference for regression coefficients is
typically based on the partial likelihood estimator together with first-order
asymptotic normal approximations
\cite{andersen1982cox,andersen1993model}.
In modern applications, however, Cox regression is frequently combined with
penalization techniques such as the Lasso \cite{tibshirani1996regression}
to handle settings with many covariates and to perform variable selection.
Penalized Cox models with a growing or moderate number of covariates
have been studied extensively in recent years
\cite{fan2002variable,tibshirani1997lasso,zhang2007adaptive}.
Theoretical guarantees for $\ell_1$-penalized estimation in the Cox model
are established in Huang et al.~\cite{huang2013lasso},
while asymptotic inference for models with a diverging number of covariates
is developed in Xia et al.~\cite{xia2023coxdiverging}
and Yu et al.~\cite{yu2021coxci}.

While $\ell_1$-penalized Cox estimators are effective for prediction and model
selection, valid statistical inference after variable selection remains
challenging.
Naive Wald-type confidence intervals applied to selected coefficients are well
known to suffer from substantial coverage distortions, as the selection step
alters the sampling distribution of the estimator
\cite{LeebPotscher2005,LeebPotscher2006,Berk2013PostSelection}.
These effects are particularly pronounced in finite samples and in the presence
of censoring, which is intrinsic to time-to-event data.

Recent methodological developments address these challenges by constructing
debiased or desparsified estimators
\cite{vandegeer2014optimal,zhang2014confidence}
and adapting these ideas to Cox regression
\cite{xia2023coxdiverging,kong2021robustcox}.
These estimators admit an asymptotically linear representation and enable
first-order valid inference for low-dimensional components of the regression
parameter in regularized Cox regression settings.
However, inference based solely on asymptotic normal approximations may still
exhibit non-negligible finite-sample distortions, especially under moderate
sample sizes, strong covariate correlations, or heavy censoring.

Bootstrap methods are a standard tool for refining inference beyond
first-order asymptotics
\cite{efron1993bootstrap,davison1997bootstrap,hall1992bootstrap}.
However, once variable selection is involved, their theoretical validity
is no longer guaranteed.
For $\ell_1$-penalized estimators, the limiting distribution is non-smooth,
and classical bootstrap procedures are known to fail
\cite{knight2000asymptotics,chatterjee2011bootstrapping}.
This has contributed to the widespread view that bootstrap-based inference
after regularized variable selection is fundamentally unreliable.

In Cox regression, the situation appears even more delicate due to
censoring and the semiparametric structure of the partial likelihood
\cite{hjort1985,spiekerman1998}.
Taken together, these considerations suggest that bootstrap methods
should be treated with caution in regularized survival settings.

\begingroup
Importantly, these negative results do not apply to debiased estimators, which
behave asymptotically like smooth, regular estimators.
This opens the door to bootstrap procedures applied to debiased Cox estimators,
rather than to the original penalized estimators themselves.
In particular, weighted resampling schemes such as Efron's bootstrap
\cite{efron1993bootstrap}
and the wild bootstrap
\cite{lin1993,spiekerman1998,vanderVaartWellner1996}
naturally respect the martingale structure underlying the
partial likelihood score and avoid resampling of censoring times or risk sets.

In this manuscript, we investigate bootstrap-based inference for the debiased Cox
estimator after Cox Lasso variable selection.
We focus on the wild bootstrap and Efron's bootstrap applied to the debiased estimator
and study their performance for post-selection inference.
Rather than proposing an exact selective inference procedure
\cite{lee2016exact,taylor2018post},
our goal is to assess whether bootstrap refinement of debiased Cox inference provides a
practical and reliable approach for constructing confidence intervals for
selected regression coefficients in finite samples.

We study the proposed methods from both a theoretical and an empirical
perspective.
Simulation studies illustrate their finite-sample behavior under a range of
survival-analysis scenarios \cite{morris2019using},
and an illustrative data example
demonstrates their practical applicability.

The remainder of the manuscript is organized as follows.
Section~\ref{sec:methods} introduces the Cox model, the penalized Cox estimator,
and the debiased Cox estimator.
Section~\ref{sec:bootstrap} presents the wild bootstrap, Efron's bootstrap,
and the resulting confidence interval constructions.
Section~\ref{sec:asymptotics} gives the asymptotic justification of the proposed procedures.
Section~\ref{sec:simstudy} describes the simulation study and evaluates their finite-sample performance.
Section~\ref{sec:example} presents the real-data example, and
Section~\ref{sec:discussion} concludes.\endgroup

\section{Methods}
\label{sec:methods}

This section introduces the Cox model, the penalized Cox estimator, and the debiased estimator that forms the basis for the
bootstrap procedures developed in the following section.

\subsection{Data structure and Cox model}
\label{sec:model}

\begingroup
We keep the manuscript notation at the level needed to define the estimator
and the bootstrap procedures. Proof-specific empirical-process notation,
detailed risk-set sums, and extended regularity arguments are collected in
the Supplementary Material. This separates the standard Cox model setup from the
methodological contribution, namely bootstrap calibration of the debiased
Cox estimator.

We introduce the statistical framework underlying the Cox proportional hazards
model under right censoring
\cite{cox1972regression,andersen1982cox,andersen1993model,therneau2000cox}.

In the manuscript, we use only the notation required for the estimator and the
bootstrap algorithm. For a vector $\mathbf{x}\in\mathbb{R}^p$, let
$x^{(j)}$ denote its $j$th component and let $(\cdot)^\top$ denote
transposition. We write $\mathbbm{1}(\cdot)$ for the indicator function.
Norms, tensor powers, empirical risk-set sums, and the filtration notation used
in the martingale arguments are given in the Supplementary Material.\endgroup

Let $T_i$ denote the event time of interest and $C_i$ the censoring time.
We observe
\[
\tilde T_i=\min(T_i,C_i),
\qquad
\delta_i=\mathbbm{1}(T_i\le C_i),
\]
together with a $p$-dimensional covariate vector
\[
\mathbf{X}_i
=
(X_i^{(1)},\ldots,X_i^{(p)})^\top
\in\mathbb{R}^p.
\]

The observed data consist of independent and identically distributed
random variables
\[
O_i=(\tilde T_i,\delta_i,\mathbf{X}_i),
\qquad i=1,\ldots,n.
\]

We assume non-informative censoring in the sense that,
conditional on $\mathbf{X}_i$, the censoring time $C_i$ is independent of
the event time $T_i$ and does not depend on the regression parameter
\cite{andersen1993model,therneau2000cox}.

Define the counting process and at-risk indicator by
\[
N_i(t)=\mathbbm{1}(\tilde T_i\le t,\delta_i=1),
\qquad
Y_i(t)=\mathbbm{1}(\tilde T_i\ge t),
\qquad t\ge 0,
\]
as in the counting-process formulation of the Cox model
\cite{andersen1982cox,andersen1993model}.

We consider the Cox proportional hazards model
\cite{cox1972regression,andersen1982cox,therneau2000cox},
which specifies the conditional hazard function as
\[
\lambda(t\mid \mathbf{X}_i)
=
\lambda_0(t)\exp(\mathbf{X}_i^\top\boldsymbol{\beta}),
\]
where $\lambda_0(t)$ is an unknown baseline hazard function and
$\boldsymbol{\beta}\in\mathbb{R}^p$ is the regression parameter vector.
We denote its true value by $\boldsymbol{\beta}^0$.

The corresponding conditional survival function is
\[
S(t\mid \mathbf{X}_i)
=
\exp\!\left\{
-\Lambda_0(t)\exp(\mathbf{X}_i^\top\boldsymbol{\beta})
\right\},
\]
where
\[
\Lambda_0(t)=\int_0^t\lambda_0(u)\,du
\]
denotes the baseline cumulative hazard
\cite{cox1972regression,therneau2000cox}.

Estimation of $\boldsymbol{\beta}$ is based on the Cox partial likelihood
\cite{cox1972regression,andersen1982cox,therneau2000cox}.
The log partial likelihood is
\begin{equation}
\label{eq:partiallik}
\ell(\boldsymbol{\beta})
=
\sum_{i=1}^n \delta_i
\left(
\mathbf{X}_i^\top\boldsymbol{\beta}
-
\log
\sum_{j:\tilde T_j\ge \tilde T_i}
\exp(\mathbf{X}_j^\top\boldsymbol{\beta})
\right).
\end{equation}

The score function can be written as
\[
\mathbf{U}(\boldsymbol{\beta})
=
\nabla_{\boldsymbol{\beta}}\ell(\boldsymbol{\beta})
=
\sum_{i=1}^n \boldsymbol{\psi}_i(\boldsymbol{\beta}),
\]
where $\nabla_{\boldsymbol{\beta}}$ denotes the gradient with respect to
$\boldsymbol{\beta}$. The individual score contributions have the
counting-process representation
\cite{andersen1982cox,andersen1993model}
\[
\boldsymbol{\psi}_i(\boldsymbol{\beta})
=
\int_0^\tau
\left\{
\mathbf{X}_i-\bar{\mathbf{X}}(t,\boldsymbol{\beta})
\right\}
\,dN_i(t),
\]
where
\[
\bar{\mathbf{X}}(t,\boldsymbol{\beta})
=
\frac{
\sum_{j=1}^n
Y_j(t)\mathbf{X}_j
\exp(\mathbf{X}_j^\top\boldsymbol{\beta})
}{
\sum_{j=1}^n
Y_j(t)
\exp(\mathbf{X}_j^\top\boldsymbol{\beta})
}.
\]

For the methodological construction below, the individual score contributions
$\boldsymbol{\psi}_i(\boldsymbol{\beta})$ are the essential objects.
The full counting-process machinery, including the predictable integrand
$\boldsymbol{\phi}_i(t,\boldsymbol{\beta})$, the Doob--Meyer decomposition,
the filtration, and the explicit martingale representation of
$\mathbf{U}(\boldsymbol{\beta}^0)$, is collected in the Supplementary Material.
In the manuscript, we only use the resulting fact that the Cox score is based on  martingale increments, which motivates multiplier resampling of the score
contributions.
\subsection{Penalized Cox regression}

We consider a penalized Cox regression setting in which the number of
covariates may be non-negligible relative to the sample size. In this setting,
unpenalized partial-likelihood estimation can be unstable, and regularization
is useful; see, for example,
\cite{tibshirani1997lasso,fan2002variable,huang2013lasso}.

To stabilize estimation and exploit sparsity in the underlying regression
structure, we use an $\ell_1$-penalized Cox estimator. Specifically,
\[
\hat{\boldsymbol{\beta}}
=
\arg\max_{\boldsymbol{\beta}\in\mathbb{R}^p}
\left\{
\ell(\boldsymbol{\beta})-\gamma_n\|\boldsymbol{\beta}\|_1
\right\},
\]
where $\gamma_n>0$ is a tuning parameter whose theoretical rate
requirements are specified in the Supplementary Material.

This is the Cox Lasso estimator introduced by
Tibshirani~\cite{tibshirani1997lasso}; see also
\cite{fan2002variable,zhang2007adaptive,hastie2015statistical}.

The $\ell_1$ penalty promotes sparse solutions by shrinking small
coefficients toward zero and setting some coefficients exactly to zero,
thereby combining estimation and variable selection.

The active set selected by the penalized estimator is
\[
\widehat M=\{j:\hat\beta_j\neq 0\}.
\]

Inference below is reported componentwise for coordinates $j\in\widehat M$.

\begingroup
\subsection{Debiased Cox estimator}

Although the $\ell_1$-penalized Cox estimator is well suited for sparse
regularized estimation, the shrinkage induced by the penalty introduces bias
that complicates direct distributional inference for the regression
coefficients
\cite{tibshirani1996regression,knight2000asymptotics,hastie2015statistical}.
One way to address this issue is to construct a debiased or desparsified
estimator; see \cite{zhang2014confidence,vandegeer2014optimal} for the general
framework and \cite{yu2021coxci,kong2021robustcox,xia2023coxdiverging}
for corresponding developments in the Cox model.

Let
\[
\ell_n(\boldsymbol{\beta})
=
n^{-1}\ell(\boldsymbol{\beta}),
\qquad
\boldsymbol{\dot\ell}_n(\boldsymbol{\beta})
=
\nabla_{\boldsymbol{\beta}}\ell_n(\boldsymbol{\beta})
=
n^{-1}\mathbf{U}(\boldsymbol{\beta})
\]
denote the normalized log partial likelihood and its score, respectively.

Let $\Sigma_0$ denote the population information matrix appearing in the
linear expansion of the Cox score and let
\[
\Theta_0=\Sigma_0^{-1}.
\]
In practice, $\Theta_0$ is replaced by a sparse estimator $\hat\Theta$
constructed by nodewise Lasso from the original data. The detailed
construction and corresponding rate conditions are given in the online
supplement. Once estimated, $\hat\Theta$ is kept fixed throughout the
bootstrap replications.

Following the general debiasing strategy
\cite{zhang2014confidence,vandegeer2014optimal}
and its adaptations to the Cox model
\cite{yu2021coxci,kong2021robustcox,xia2023coxdiverging},
we define the debiased estimator by the one-step correction
\[
\tilde{\boldsymbol{\beta}}
=
\hat{\boldsymbol{\beta}}
+
\hat\Theta\,
\boldsymbol{\dot\ell}_n(\hat{\boldsymbol{\beta}}).
\]
For each $j\in\widehat M$, inference is based on the debiased estimator
$\tilde\beta_j$ and targets the corresponding true Cox regression
coefficient $\beta_j^0$. Thus, $\widehat M$ determines which coefficients
are reported, while the inferential target remains $\beta_j^0$.

The one-step correction reduces the leading shrinkage bias of the Cox Lasso
estimator. Under suitable regularity and sparsity conditions,
$\tilde\beta_j$ admits an asymptotically linear representation and is
asymptotically normal around $\beta_j^0$. This provides the basis for
componentwise confidence intervals.

In this manuscript, we use bootstrap resampling to approximate the sampling
distribution of the debiased estimator. The bootstrap procedures exploit its
asymptotically linear representation and resample the leading score
contribution rather than the non-smooth Cox Lasso estimator itself.
\endgroup

\begingroup
\section{Bootstrap inference}
\label{sec:bootstrap}

We consider two score-based bootstrap procedures for the debiased Cox
estimator: the wild bootstrap and Efron's bootstrap. Both procedures are
applied after the Cox Lasso estimator $\hat{\boldsymbol{\beta}}$ and the
inverse-information estimator $\hat\Theta$ have been estimated from the
original data, with $\hat\Theta$ obtained by nodewise Lasso. Neither quantity
is re-estimated within the bootstrap. Instead, each bootstrap replication
generates a weighted version of the normalized score, which is then multiplied
by the fixed matrix $\hat\Theta$.

The two bootstrap procedures differ in the weights applied to the individual
score contributions. The wild bootstrap uses independent mean-zero
multipliers, whereas Efron's bootstrap uses centered multinomial weights.
In both cases, the weighted score is used to generate bootstrap replicates of
the debiased estimator. Thus, the bootstrap approximates the first-order
sampling variation of the debiased estimator without repeating the Cox Lasso
estimation or variable-selection step in each bootstrap replication.\endgroup

\subsection{Wild bootstrap}
\label{sec:wild}

The wild bootstrap exploits the counting-process and martingale structure
of the Cox model
\cite{andersen1982cox,andersen1993model,therneau2000cox}.
In the Andersen--Gill formulation, the Cox score admits a martingale
representation, which motivates multiplier-based resampling for censored
survival data; see, for example,
\cite{lin1993,spiekerman1998,vanderVaartWellner1996,kosorok2008empirical,
dobler2019confidence,bluhmki2019wild}.

For bootstrap replication $b=1,\ldots,B$, let
$\xi_1^{(b)},\ldots,\xi_n^{(b)}$ be i.i.d.\ multiplier variables satisfying
\[
\mathbb{E}(\xi_i^{(b)})=0,
\qquad
\mathbb{E}\{(\xi_i^{(b)})^2\}=1,
\qquad
\mathbb{E}\{(\xi_i^{(b)})^4\}<\infty.
\]
Such moment conditions are standard in multiplier bootstrap theory
\cite{vanderVaartWellner1996,kosorok2008empirical,chernozhukov2013gaussian}.
Possible choices include Rademacher, Gaussian, Mammen
\cite{mammen1993}, and centered exponential multipliers
\cite{beyersmann2013,dobler2019confidence,bluhmki2019wild}.

The wild bootstrap perturbs the individual score contributions evaluated at
the Cox Lasso estimator $\hat{\boldsymbol{\beta}}$. For bootstrap replication
$b$, define
\[
\boldsymbol{\dot\ell}_{n,\mathrm{WB}}^{*(b)}
(\hat{\boldsymbol{\beta}})
=
\frac{1}{n}
\sum_{i=1}^n
\xi_i^{(b)}
\boldsymbol{\psi}_i(\hat{\boldsymbol{\beta}}).
\]

This construction corresponds to the leading term in the asymptotically
linear representation
\[
\tilde{\boldsymbol{\beta}}-\boldsymbol{\beta}^0
=
\Theta_0
\frac{1}{n}
\sum_{i=1}^n
\boldsymbol{\psi}_i(\boldsymbol{\beta}^0)
+
o_p(n^{-1/2}),
\]
or equivalently,
\[
\sqrt n
(\tilde{\boldsymbol{\beta}}-\boldsymbol{\beta}^0)
=
\Theta_0
\frac{1}{\sqrt n}
\sum_{i=1}^n
\boldsymbol{\psi}_i(\boldsymbol{\beta}^0)
+
o_p(1).
\]

The wild bootstrap replicate of the debiased estimator is therefore defined
as
\[
\tilde{\boldsymbol{\beta}}_{\mathrm{WB}}^{*(b)}
=
\tilde{\boldsymbol{\beta}}
+
\hat\Theta\,
\boldsymbol{\dot\ell}_{n,\mathrm{WB}}^{*(b)}
(\hat{\boldsymbol{\beta}}).
\]
Equivalently,
\[
\sqrt n
\left(
\tilde{\boldsymbol{\beta}}_{\mathrm{WB}}^{*(b)}
-
\tilde{\boldsymbol{\beta}}
\right)
=
\hat\Theta
\frac{1}{\sqrt n}
\sum_{i=1}^n
\xi_i^{(b)}
\boldsymbol{\psi}_i(\hat{\boldsymbol{\beta}}).
\]

\begingroup
Thus, each wild bootstrap replication consists of weighting the estimated
individual score contributions and applying the fixed debiasing matrix
$\hat\Theta$. The Cox Lasso estimator and the selected set are not
re-estimated within the bootstrap. By resampling only the leading linear
score contribution, the procedure leaves the observed risk sets unchanged
and remains aligned with the martingale structure underlying the Cox score
\cite{lin1993,spiekerman1998,dobler2019confidence}.\endgroup

\begingroup
\subsection{Efron's bootstrap}
\label{sec:Efron}

In the classical Efron bootstrap, a bootstrap sample of size \(n\) is obtained
by drawing observations with replacement from the original data. Equivalently,
if \(W_i\) denotes the number of times observation \(i\) is selected, the
resulting vector of selection counts follows
\[
(W_1,\ldots,W_n)
\sim
\mathrm{Mult}\{n;1/n,\ldots,1/n\}.
\]
Thus, multinomial counts provide a weighted representation of classical
resampling with replacement.

In our score-based implementation, we use the same multinomial resampling law
without refitting the Cox Lasso in each bootstrap sample. Instead, the centered
counts are used to reweight the individual score contributions, while the Cox
Lasso estimator \(\hat{\boldsymbol{\beta}}\) and the debiasing matrix
\(\hat\Theta\), both obtained from the original data, are kept fixed. This
yields an exchangeably weighted bootstrap; see
\cite{pauly2011weighted} for a general treatment of weighted resampling with
exchangeable weights and \cite{doblerpauly2014} for related survival
resampling formulations.

The difference from the wild bootstrap is therefore the distribution of the
bootstrap weights. The wild bootstrap uses independent mean-zero multipliers,
whereas Efron's bootstrap uses centered multinomial counts, which are
exchangeable but dependent.

Specifically, for bootstrap replication \(b=1,\ldots,B\), draw
\[
(W_1^{(b)},\ldots,W_n^{(b)})
\sim
\mathrm{Mult}\{n;1/n,\ldots,1/n\},
\]
and define the centered weights
\[
\omega_i^{(b)}
=
W_i^{(b)}-1,
\qquad
i=1,\ldots,n.
\]

Using these weights, we define the normalized score for Efron's bootstrap by
\[
\boldsymbol{\dot\ell}_{n,\mathrm{Efr}}^{*(b)}
(\hat{\boldsymbol{\beta}})
=
\frac{1}{n}
\sum_{i=1}^n
\omega_i^{(b)}
\,\boldsymbol{\psi}_i(\hat{\boldsymbol{\beta}}).
\]

The corresponding debiased bootstrap estimator is
\[
\tilde{\boldsymbol{\beta}}_{\mathrm{Efr}}^{*(b)}
=
\tilde{\boldsymbol{\beta}}
+
\hat\Theta\,
\boldsymbol{\dot\ell}_{n,\mathrm{Efr}}^{*(b)}
(\hat{\boldsymbol{\beta}}).
\]

Equivalently, on the \(\sqrt n\) scale,
\[
\sqrt n
\left(
\tilde{\boldsymbol{\beta}}_{\mathrm{Efr}}^{*(b)}
-
\tilde{\boldsymbol{\beta}}
\right)
=
\hat\Theta
\frac{1}{\sqrt n}
\sum_{i=1}^n
\omega_i^{(b)}
\,\boldsymbol{\psi}_i(\hat{\boldsymbol{\beta}}).
\]

Thus, Efron's bootstrap uses the classical resampling-with-replacement law at
the level of the linear score representation of the debiased estimator,
without repeating the Cox Lasso estimation or variable-selection step within
the bootstrap.
\endgroup

\subsection{Confidence interval construction}
\label{sec:ci}

We now describe the construction of componentwise confidence intervals
based on bootstrap replicates of the debiased estimator. The formulas
below apply to both bootstrap schemes considered in this manuscript, namely
the wild bootstrap and Efron's bootstrap. To keep the notation concise,
we write $(*)$ for either resampling scheme.

Let $B$ denote the number of bootstrap replications and
$\alpha\in(0,1)$ the nominal significance level. Inference is conducted
componentwise for $j\in\widehat M$.

\begingroup
The intervals should be interpreted as intervals for reported debiased coordinates after the original selection step, not as exact selective-inference intervals for the event $\{\widehat M=M\}$. This is consistent with the construction above, where the bootstrap calibrates the distributional approximation of the debiased estimator and does not model support variability inside the bootstrap loop.\endgroup

For a fixed coordinate $j$, let
\[
\tilde\beta_j^{*(1)},\ldots,\tilde\beta_j^{*(B)}
\]
denote the bootstrap replicates of the debiased estimator
$\tilde\beta_j$, where $(*)$ stands for either the wild bootstrap or
Efron's bootstrap.

Define the bootstrap deviations by
\[
\Delta_j^{*(b)}
=
\tilde\beta_j^{*(b)}-\tilde\beta_j,
\qquad b=1,\ldots,B,
\]
and the empirical bootstrap standard deviation by
\[
\hat\sigma_{j,*}
=
\left\{
\frac{1}{B-1}
\sum_{b=1}^B
\big(
\tilde\beta_j^{*(b)}-\bar\beta_j^*
\big)^2
\right\}^{1/2},
\]
where
\[
\bar\beta_j^*
=
\frac{1}{B}
\sum_{b=1}^B
\tilde\beta_j^{*(b)}
\]
denotes the bootstrap mean.

For any finite sample $\{V^{(1)},\ldots,V^{(B)}\}$, let
$\hat Q_\gamma(V)$ denote the empirical $\gamma$-quantile.

\subsubsection*{Percentile interval}

Let
\[
\hat Q_{j,\gamma}^*
=
\hat Q_\gamma
\big(
\tilde\beta_j^{*(1)},\ldots,\tilde\beta_j^{*(B)}
\big).
\]
The percentile interval is defined as
\[
\CI_j^{\mathrm{perc},*}(1-\alpha)
=
\big[
\hat Q_{j,\alpha/2}^*,
\hat Q_{j,1-\alpha/2}^*
\big].
\]

\subsubsection*{Bootstrap normal interval}

The bootstrap normal interval is given by
\[
\CI_j^{\mathrm{norm},*}(1-\alpha)
=
\tilde\beta_j
\pm
z_{1-\alpha/2}\hat\sigma_{j,*},
\]
where $z_{1-\alpha/2}$ denotes the $(1-\alpha/2)$-quantile of the
standard normal distribution.

\begingroup
\subsubsection*{Basic bootstrap interval}

Let
\[
\hat Q_{j,\gamma}^{\Delta,*}
=
\hat Q_\gamma
\big(
\Delta_j^{*(1)},\ldots,\Delta_j^{*(B)}
\big)
\]
denote the empirical $\gamma$-quantile of the bootstrap deviations. The
basic, or reverse-percentile, bootstrap interval
\cite{efron1993bootstrap,davison1997bootstrap} is
\[
\CI_j^{\mathrm{basic},*}(1-\alpha)
=
\big[
\tilde\beta_j
-
\hat Q_{j,1-\alpha/2}^{\Delta,*},
\;
\tilde\beta_j
-
\hat Q_{j,\alpha/2}^{\Delta,*}
\big].
\]

For the asymptotic argument below, we additionally consider the bootstrap
deviations on a standardized scale. Define
\[
T_j^{*(b)}
=
\frac{\Delta_j^{*(b)}}{\hat\sigma_{j,*}},
\qquad b=1,\ldots,B,
\]
where $\hat\sigma_{j,*}$ denotes the empirical bootstrap standard deviation,
which is common to all bootstrap replications. Consequently,
\[
\hat Q_\gamma(T_j^*)
=
\frac{\hat Q_{j,\gamma}^{\Delta,*}}{\hat\sigma_{j,*}}.
\]
Since the same scale factor is used for every bootstrap replication, it
cancels when the quantiles are transformed back to the scale of
$\tilde\beta_j$. Thus, inversion of the standardized bootstrap quantiles
yields exactly the basic interval defined above. The standardized statistic
is therefore used only as a device in the theoretical argument and does not
define a separate confidence-interval construction.

This should be distinguished from the classical bootstrap-$t$ interval,
which uses a replicate-specific standard-error estimate within each bootstrap
replication. Section~\ref{sec:asymptotics} establishes consistency of the
standardized bootstrap statistic and of the corresponding bootstrap variance
estimator.

In the sequel, we write
$\CI_j^{\mathrm{perc,WB}}$, $\CI_j^{\mathrm{norm,WB}}$,
and $\CI_j^{\mathrm{basic,WB}}$ for confidence intervals based on the
wild bootstrap, and
$\CI_j^{\mathrm{perc,Efr}}$, $\CI_j^{\mathrm{norm,Efr}}$,
and $\CI_j^{\mathrm{basic,Efr}}$ for confidence intervals based on
Efron's bootstrap.
\endgroup

\begingroup
\section{Asymptotic justification}
\label{sec:asymptotics}

This section establishes the asymptotic validity of the bootstrap
confidence intervals introduced in Section~\ref{sec:ci}.
Inference is conducted componentwise for fixed
$j\in\widehat M$.

The argument proceeds in three steps.
First, we recall asymptotic linearity and asymptotic normality
of the debiased Cox estimator; see
\cite{yu2021coxci,kong2021robustcox,xia2023coxdiverging}.
Second, we show that the conditional wild bootstrap distribution
consistently approximates the sampling distribution.
Third, we combine distributional consistency with variance consistency
to establish validity of the standardized bootstrap statistic and, through
quantile inversion, of the corresponding basic bootstrap interval.

We present these arguments in detail for the wild bootstrap.
The corresponding results for Efron's bootstrap follow from the same
linearization of the debiased Cox estimator, with the conditional
multiplier argument replaced by an exchangeably weighted bootstrap
argument for the centered multinomial weights. The detailed assumptions
and proofs for both procedures are given in Sections~S1--S2 of the
Supplementary Material.

The proof strategy separates the arguments for the Cox model from the
resampling arguments. Debiasing arguments for the Cox model first provide
an asymptotically linear representation with a negligible remainder.
The resulting leading term is a sum of influence contributions.
For the wild bootstrap, a conditional multiplier central limit theorem
is applied to the independently weighted leading term. For Efron's
bootstrap, the corresponding conditional approximation is established
using exchangeably weighted bootstrap theory, based on the result of
Pauly~\cite{pauly2011weighted}. Thus, the part of the argument specific
to the Cox model is common to both procedures, whereas the final
conditional central limit argument depends on the bootstrap weights.

Throughout this section, $P^*(\cdot)$ denotes probability conditional on
the observed data, that is, probability with respect to the bootstrap
weights only.
The proof-specific notation $\E^*$, $\Var^*$, and $o_p^*(1)$ is defined
and used only in the Supplementary Material.

The results are stated under assumptions (A1)--(A9) detailed in the
Supplementary Material.
Assumptions (A1)--(A7) cover the sampling and censoring setting,
regularity conditions for the Cox model, sparsity, precision-matrix
estimation, and moment conditions for the influence components.
Assumption (A8) specifies the wild bootstrap multipliers, whereas
(A9) specifies the centered multinomial weights used in Efron's bootstrap.\endgroup

\begingroup
\begin{remark}[Tuning-parameter selection]
The theory is formulated for the Cox Lasso tuning parameter $\gamma_n$
such that the resulting Cox Lasso estimator is sufficiently sparse and
satisfies the corresponding debiasing remainder bounds.
Thus, the proofs do not rely on a particular implementation rule such
as AIC, BIC, cross-validation at $\lambda_{\min}$, or cross-validation
at $\lambda_{1se}$, but on the resulting properties of the fitted Cox Lasso
estimator; see, for example, \cite{vandegeer2014optimal} and the
implementations for the Cox model in
\cite{yu2021coxci,kong2021robustcox}.

The notation $\lambda_{\min}$ and $\lambda_{1se}$ is retained for the
conventional cross-validation rules used in implementation.
Choosing the penalty by AIC, BIC, $\lambda_{\min}$, or $\lambda_{1se}$
does not change the $\ell_1$-penalized nature of the Cox Lasso estimator.
However, the present theory covers these choices only insofar as the
resulting random tuning parameter satisfies the regularity conditions
required for the Lasso and debiasing arguments with high probability.

The four tuning rules are therefore treated as implementable
finite-sample choices rather than as theoretically ordered procedures.
The simulation study considers $\lambda_{\min}$, $\lambda_{1se}$, AIC,
and BIC under both the standard Cox Lasso and an adaptive Cox Lasso
specification. These comparisons are sensitivity analyses for data-driven
penalty selection. The theoretical results apply only insofar as the resulting
penalized estimator satisfies the sparsity and debiasing remainder conditions
stated above.
\end{remark}\endgroup

\begin{theorem}[Asymptotic linearity]
\label{thm:asymptotic_linearity}
Assume (A1)--(A7) in the Supplementary Material.
For each fixed $j\in\widehat M$,
\[
\sqrt n(\tilde\beta_j-\beta_j^0)
=
\mathbf{e}_j^\top\Theta_0
\frac{1}{\sqrt n}
\sum_{i=1}^n
\boldsymbol{\psi}_i(\boldsymbol{\beta}^0)
+o_p(1)
=
\frac{1}{\sqrt n}\sum_{i=1}^n\varphi_{ij}
+o_p(1),
\]
where
\[
\varphi_{ij}
=
\mathbf{e}_j^\top
\Theta_0
\boldsymbol{\psi}_i(\boldsymbol{\beta}^0).
\]
Consequently,
\[
\sqrt n(\tilde\beta_j-\beta_j^0)
\Rightarrow
\mathcal N(0,\sigma_j^2),
\qquad
\sigma_j^2=\Var(\varphi_{1j}).
\]
\end{theorem}

The theorem implies that the debiased estimator behaves, to first order,
like an average of i.i.d.\ influence components.
Consequently, classical central limit theory applies.

\begin{theorem}[Wild bootstrap consistency]
\label{thm:wb-cons}
Assume (A1)--(A8) in the Supplementary Material.
For each fixed $j\in\widehat M$,
\[
\sup_{t\in\mathbb R}
\left|
P^*\!\left(
\sqrt n(\tilde\beta_j^{\mathrm{WB}}-\tilde\beta_j)\le t
\right)
-
P\!\left(
\sqrt n(\tilde\beta_j-\beta_j^0)\le t
\right)
\right|
\xrightarrow{p}0.
\]
\end{theorem}

Thus, the conditional distribution of the wild bootstrap statistic
consistently approximates the sampling distribution of
$\sqrt n(\tilde\beta_j-\beta_j^0)$.

\noindent\textit{Heuristic argument.}
By Theorem~\ref{thm:asymptotic_linearity},
$\sqrt n(\tilde\beta_j-\beta_j^0)$ is asymptotically equivalent to
the linear term
\[
\frac{1}{\sqrt n}\sum_{i=1}^n\varphi_{ij}.
\]
Conditionally on the observed data, the wild bootstrap statistic is,
up to an asymptotically negligible remainder, represented by the
corresponding multiplier sum
\[
\frac{1}{\sqrt n}\sum_{i=1}^n\xi_i\varphi_{ij}.
\]
Under (A7)--(A8), a conditional multiplier central limit theorem yields
\[
\frac{1}{\sqrt n}\sum_{i=1}^n\xi_i\varphi_{ij}
\Rightarrow^*
\mathcal N(0,\sigma_j^2)
\quad\text{in probability},
\]
while the ordinary central limit theorem gives the same limiting
distribution for
$n^{-1/2}\sum_{i=1}^n\varphi_{ij}$.
Conditional Slutsky arguments then transfer this convergence from the
linear terms to the full statistics, yielding the stated
Kolmogorov-distance result.
The rigorous argument is given in the Supplementary Material.
For wild bootstrap arguments in multiplicative intensity models, see
\cite{bluhmki2019wild}; general multiplier central limit results are
given, for example, in
\cite{vanderVaartWellner1996,kosorok2008empirical}.

\begin{theorem}[Wild bootstrap variance consistency]
\label{thm:wb-var}
Assume (A1)--(A8) in the Supplementary Material and $B\to\infty$.
Then, for each fixed $j\in\widehat M$,
\[
\hat\sigma_{j,\mathrm{WB}}^2
\xrightarrow{p}
\sigma_j^2/n.
\]
\end{theorem}

Thus, the empirical bootstrap variance consistently estimates the
asymptotic variance of $\tilde\beta_j$.

Define
\[
T_j
=
\frac{\sqrt n(\tilde\beta_j-\beta_j^0)}{\sigma_j},
\qquad
T_j^{\mathrm{WB}}
=
\frac{
\sqrt n(\tilde\beta_j^{\mathrm{WB}}-\tilde\beta_j)
}{
\sqrt n\,\hat\sigma_{j,\mathrm{WB}}
}
=
\frac{
\tilde\beta_j^{\mathrm{WB}}-\tilde\beta_j
}{
\hat\sigma_{j,\mathrm{WB}}
}.
\]

\begingroup
\begin{theorem}[Validity of the standardized wild bootstrap statistic]
\label{thm:stud-wb}
Assume (A1)--(A8) in the Supplementary Material and $B\to\infty$.
For each fixed $j\in\widehat M$,
\[
\sup_{t\in\mathbb R}
\left|
P^*\!\left(T_j^{\mathrm{WB}}\le t\right)
-
P\!\left(T_j\le t\right)
\right|
\xrightarrow{p}0.
\]
\end{theorem}

Together with the asymptotic normality of $T_j$, this establishes validity
of the standardized bootstrap statistic. Because the common scale
$\hat\sigma_{j,\mathrm{WB}}$ cancels under quantile inversion, the corresponding
confidence interval is the basic bootstrap interval defined in Section~\ref{sec:ci}.

Analogous consistency, variance-consistency, and validity results for the
standardized bootstrap statistic hold for Efron's bootstrap under assumptions (A1)--(A7) and
(A9). The Cox model linearization is the same as above, while the
independent multiplier argument is replaced by an exchangeably weighted
bootstrap argument for the centered multinomial weights. The detailed
proofs are given in Section~S2 of the Supplementary Material, where the
conditional approximation is established using the weighted-resampling
result of Pauly~\cite{pauly2011weighted}.\endgroup

\section{Simulation study}
\label{sec:simstudy}

This section investigates the finite-sample performance of bootstrap inference
for the debiased Cox estimator after Cox Lasso variable selection. We compare
the wild bootstrap and Efron's bootstrap with non-bootstrap inference in terms
of coverage accuracy, interval length, type~I error, and computational cost
across a range of data-generating settings.

\subsection{Design}
\label{sec:design}

\begingroup
The simulation design follows the ADEMP framework \cite{morris2019using},
which separates aims, data-generating mechanisms, estimands, methods, and
performance measures.

\subsubsection{Objectives}

The primary objective of the simulation study was to assess whether the
bootstrap-based intervals improve finite-sample coverage of the debiased
full-model coefficients relative to first-order Wald inference. Secondary
performance measures included interval width, type~I error for null
coefficients, Cox Lasso selection frequency in the original simulated data set,
and computational cost.

\subsubsection{Data-generating mechanisms}

The primary analyses focused on small and moderate sample sizes, using
\[
n \in \{30,40,60,75,80,100,125,150,175,200,250\}
\]
and \(p \in \{10,20\}\). The grid includes the intermediate values
\(n=75\) and \(n=80\), with \(n=80\) also serving as the common setting
for the tuning-rule sensitivity analyses reported in Section~S3.1.3 of
the Supplementary Material. Additional results for selected individual
simulation settings under the final tuning choices are reported in
Section~S3.1.1.

Covariates were generated from a multivariate normal distribution with mean
zero and autoregressive correlation structure
\[
\Sigma_{ij}=\rho^{|i-j|},
\qquad
\rho\in\{0,0.1\}.
\]
After generation, covariates were standardized within each simulated data set.

\begingroup
The true regression vector followed one of four coefficient patterns:
\textit{allones}, \textit{highcontrast}, \textit{realistic}, and
\textit{sparse}. For the \textit{allones} pattern, all coefficients were
set to one. For \textit{highcontrast}, the first four coefficients were
$(1,-1,1,-1)$; for \textit{sparse}, they were $(1,0,1,0)$; and for
\textit{realistic}, they were $(0.8,0.7,0.5,0.8)$.
Except for the \textit{allones} pattern, all remaining coefficients were
set to zero.

In the manuscript and the supplementary sensitivity analyses
in Section~S3 of the Supplementary Material, we focus on the realistic
pattern.\endgroup

Event times were generated from a Cox proportional hazards model with
linear predictor \(\mathbf{X}^\top\boldsymbol{\beta}^0\). Conditional on
\(\mathbf{X}\),
\[
\Pr(T>t\mid \mathbf{X})
=
\exp\left\{
-H_0(t)\exp(\mathbf{X}^\top\boldsymbol{\beta}^0)
\right\}.
\]
Using inverse-transform sampling, this yields
\[
T
=
H_0^{-1}\left(
\frac{-\log U}{\exp(\mathbf{X}^\top\boldsymbol{\beta}^0)}
\right),
\qquad
U\sim\mathrm{Unif}(0,1).
\]

We considered an exponential baseline hazard with \(H_0(t)=t\) and a
Weibull baseline hazard with \(H_0(t)=s t^k\), where \(s=1\) and \(k=2\).

Independent right censoring times were generated from an exponential
distribution. The censoring rate was calibrated separately for each scenario
to achieve the target censoring proportion
\[
\pi_C\in\{0,0.10,0.30\},
\]
following standard simulation practice for censored survival data
\cite{ramos2020sampling}.

The observed time and event indicator were
\[
\tilde T = \min(T,C), \qquad
\delta = \mathbbm{1}(T \le C).
\]
Small deterministic perturbations were added to duplicated event and censoring
times to avoid numerical problems due to ties. The perturbations were of order
\(10^{-8}\) for event times and \(5\cdot10^{-9}\) for censoring times, scaled
by the range of the respective time variable.

\subsubsection{Variable selection and tuning}

\begingroup
For each simulated data set, variable selection was performed using either the
standard Cox Lasso or an adaptive Cox Lasso. The same four tuning rules were used for the standard and adaptive Cox Lasso analyses.

For the cross-validation-based rules, fivefold cross-validation was used.
The value \(\lambda_{\min}\) denotes the penalty parameter minimizing the
cross-validated Cox partial-likelihood deviance, whereas
\(\lambda_{1se}\) denotes the largest penalty parameter within one standard
error of the minimum. For the information-criterion-based rules, the penalty
parameter was chosen by minimizing AIC or BIC along the Cox Lasso solution
path.

For the primary analyses, variable selection was performed using the standard
Cox Lasso. As a sensitivity analysis, we additionally considered an adaptive
Cox Lasso to assess whether the finite-sample behaviour of the inferential
procedures was sensitive to the form of penalization.

Following the adaptive Lasso principle
\cite{zou2006adaptive,zhang2007adaptive}, an initial standard Cox Lasso fit was
used to construct coefficient-specific penalty weights
\[
w_j=(|\hat\beta_j|+\varepsilon)^{-\kappa},
\]
with \(\kappa=1\). The weighted Cox Lasso was then fitted using the same
data-driven tuning rule as in the corresponding standard Cox Lasso analysis.

For both the primary standard Cox Lasso analyses and the adaptive sensitivity
analyses, the regularization parameter was selected using
\(\lambda_{\min}\), \(\lambda_{1se}\), AIC, or BIC. Thus, the adaptive
analyses vary the penalty specification while retaining the same set of
data-driven tuning strategies.

The resulting eight combinations,
given by four tuning rules under standard and adaptive penalization, are
therefore interpreted as implementable finite-sample analysis strategies rather
than oracle choices. The asymptotic theory does not establish validity of any
one data-driven tuning rule in full generality, but requires the resulting
penalized estimator to satisfy the sparsity and debiasing remainder conditions
stated in Section~\ref{sec:asymptotics}.
\endgroup

\begingroup
Unless stated otherwise, the simulation summaries were based on \(M=1000\)
Monte Carlo replications, with \(B=900\) bootstrap replications for each
bootstrap procedure. 
\endgroup

\begingroup
\subsection{Estimands and performance measures}
\label{sec:estimand}

\begingroup
Let \(\widehat M_{\mathrm{orig}}\) denote the model selected by the Cox Lasso in
the original simulated data set. For inference based on the debiased Cox
estimator, including the corresponding Wald, wild bootstrap, and Efron's
bootstrap intervals, the formal target is the full-model coefficient
\(\beta_j^0\). Cox Lasso selection determines which coefficients would be
reported in an application, but it does not redefine this target.

Empirical coverage is calculated coefficient-wise over the Monte Carlo
replications for which the corresponding interval is available and is not
conditioned on the event
\[
j\in\widehat M_{\mathrm{orig}}.
\]
In the simulation study, \(\beta_j^0\) is known and serves as the reference
value for coverage.

Selection behaviour is evaluated separately from coverage. We define
\[
N_{\mathrm{sel},j}
=
\sum_{m=1}^{M}
\mathbbm{1}\{j\in\widehat M_{\mathrm{orig}}^{(m)}\},
\qquad
\widehat p_{\mathrm{sel},j}
=
\frac{N_{\mathrm{sel},j}}{M},
\]
where \(\widehat p_{\mathrm{sel},j}\) denotes the empirical Cox Lasso selection
frequency of coefficient \(j\) across the original simulated data sets.
Selection frequencies are considered separately from interval availability
for the debiased procedures.

For active coefficients, we additionally summarize mean interval width. For
inactive coefficients, type~I error is used as the corresponding
null-coefficient diagnostic.
\endgroup
\subsection{Methods compared}

\begingroup
The primary numerical comparison focuses on three inferential procedures:
\begin{enumerate}
\item Wald inference based on the debiased Cox estimator;
\item the basic interval based on the wild bootstrap and the debiased Cox
estimator;
\item the basic interval based on Efron's bootstrap and the debiased Cox
estimator.
\end{enumerate}
The oracle Cox estimator fitted on the true active set is included only as an
unattainable benchmark in selected displays. The percentile and bootstrap-normal
constructions defined in Section~\ref{sec:ci} are not part of the primary
numerical comparison reported here.

The main displays emphasize representative settings with \(\lambda_{\min}\)
and AIC for compactness. The full simulation design additionally includes
\(\lambda_{1se}\) and BIC and evaluates all four tuning rules under the standard
Cox Lasso and, as a sensitivity analysis, under the adaptive Cox Lasso. The
Supplementary Material reports the corresponding selection, coverage, and interval
width results across the four tuning rules and both penalty specifications,
evaluated under the same simulation setting. The
bootstrap procedures are applied after the debiasing step and therefore refine
the finite-sample distributional approximation of the debiased estimator rather
than bootstrapping the original non-smooth Cox Lasso estimator directly.
\endgroup
\endgroup

\begingroup
\subsection{Results}
\label{sec:results}

\begingroup
We report the finite-sample behaviour coefficient by coefficient rather than
averaging coverage over all variables. This distinction is important because
active and inactive coefficients can exhibit different finite-sample behaviour.
The main simulation summaries focus on the realistic coefficient pattern.
Active coefficients are denoted by \(\beta_1,\ldots,\beta_4\); as representative
inactive coefficients, we additionally report \(\beta_5\) and \(\beta_6\).

For the coefficient-specific coverage summaries, we use two named reference
settings throughout the manuscript and Supplementary Material. \textbf{Reference
setting R1 (standard-\(\lambda_{\min}\))} uses a Weibull baseline hazard,
\(p=10\), \(n=100\), 10\% target censoring, independent covariates
(\(\rho=0\)), and the standard Cox Lasso with \(\lambda_{\min}\) tuning.
\textbf{Reference setting R2 (adaptive-AIC)} uses the same data-generating
configuration but the adaptive Cox Lasso with AIC tuning.

Table~\ref{tab:coef-coverage-min} reports coefficient-specific empirical
coverage for Reference setting~R1. For active coefficients, the debiased Wald
interval is still below the nominal 90\% level in several cases. For example,
for \(\beta_1\), empirical coverage is 0.861 for the debiased Wald interval,
0.923 for the basic interval based on the wild bootstrap, and 0.982 for the
basic interval based on Efron's bootstrap. Thus, the bootstrap intervals improve
coverage for this coefficient, with the interval based on Efron's bootstrap
being more conservative.
\endgroup

\input{generated_tables_main/table_coverage_min_adaptfalse_p10_n100.tex}

\begingroup
Table~\ref{tab:coef-coverage-aic} reports coefficient-specific empirical
coverage for Reference setting~R2. Reference settings~R1 and~R2 are not intended
as a direct comparison of tuning rules because both the tuning rule and the
penalty specification differ. Instead, they illustrate coefficient-specific
finite-sample behaviour under two practically relevant analysis choices.
Overall, the sensitivity analyses with the adaptive Cox Lasso showed
qualitatively similar coverage patterns to the primary standard Cox Lasso
analyses.
\endgroup

\input{generated_tables_main/table_coverage_aic_adapttrue_p10_n100.tex}

\begingroup
Tables~\ref{tab:coef-coverage-min} and~\ref{tab:coef-coverage-aic} therefore
report the coefficient-specific results for Reference settings~R1 and~R2,
respectively. To complement these raw coverage results, we additionally
summarize deviations from the nominal coverage level using four named focus
settings, denoted Settings~A--D and defined below.
Figure~\ref{fig:basic-representative-error-settings} reports the results for
Settings~A--D, while Figure~\ref{fig:basic-average-error-by-n} summarizes
average absolute coverage error across sample sizes. Numerical results for
Settings~A--D are reported in Table~\ref{tab:basic-representative-settings},
and the numerical summary of average absolute coverage error by sample size is
reported in Table~\ref{tab:basic-average-error-by-n}.

More detailed simulation results are provided in Section~S3.1 of the online
supplement. Section~S3.1.1 provides additional results for individual
simulation settings and the complete summary of average absolute coverage
error for active coefficients. Section~S3.1.2 gives coefficient-specific
coverage results for Reference settings~R1 and~R2. Section~S3.1.3 compares
all eight combinations of tuning rule and penalty specification under the
same simulation setting, including Cox Lasso selection frequencies from the
original simulated data sets, mean coverage for active coefficients, and
interval width. Section~S3.1.4 provides coefficient-specific interval-width
results for Reference settings~R1 and~R2.
\endgroup

\begingroup
Because raw coverage alone can make conservative overcoverage appear
preferable, we summarize performance for active coefficients by the absolute
coverage error
\[
|\widehat{\mathrm{coverage}}-0.90|.
\]
This treats undercoverage and overcoverage as deviations from the nominal
level. Type~I error is reported separately for null coefficients.

For the graphical and numerical comparisons in
Figure~\ref{fig:basic-representative-error-settings} and
Table~\ref{tab:basic-representative-settings}, we define four focus settings
that are used consistently throughout the manuscript and Supplementary Material:
\begin{itemize}
\item \textbf{Setting A (AIC-correlated-censored active):}
      \(\beta_4\), \(p=20\), AIC tuning, target censoring
      \(\pi_C=0.10\), and \(\rho=0.1\);
\item \textbf{Setting B (AIC-independent-uncensored active):}
      \(\beta_2\), \(p=20\), AIC tuning, no censoring
      (\(\pi_C=0\)), and \(\rho=0\);
\item \textbf{Setting C (\(\lambda_{\min}\)-independent-censored active):}
      \(\beta_1\), \(p=10\), \(\lambda_{\min}\) tuning, target censoring
      \(\pi_C=0.10\), and \(\rho=0\);
\item \textbf{Setting D (AIC-correlated-censored null):}
      \(\beta_{10}\), \(p=20\), AIC tuning, target censoring
      \(\pi_C=0.10\), and \(\rho=0.1\).
\end{itemize}
The realistic coefficient pattern is used throughout Settings~A--D.
For Settings~A--C, absolute coverage error is evaluated for the indicated
active coefficient; Setting~D evaluates type~I error for the null coefficient
\(\beta_{10}\).
\endgroup

\begin{figure}[t]
  \centering
  \includegraphics[width=0.98\linewidth]
  {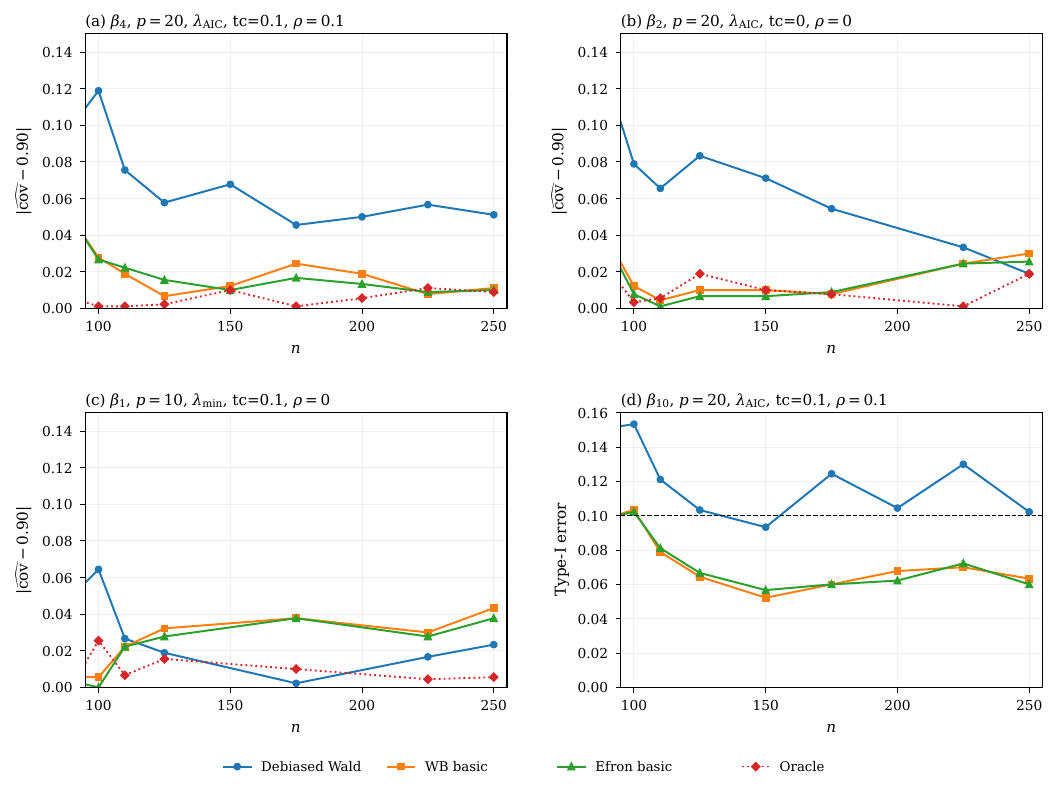}
  \caption{Finite-sample performance in focus Settings~A--D.
  Settings~A--C report absolute coverage error relative to the nominal 90\%
  level for the specified active coefficients, with smaller values indicating
  closer agreement with the nominal level. Setting~D reports type~I error for
  the null coefficient \(\beta_{10}\). The debiased Wald interval is compared
  with the basic intervals based on the wild bootstrap and Efron's bootstrap.}
  \label{fig:basic-representative-error-settings}
\end{figure}

\begingroup
Figure~\ref{fig:basic-representative-error-settings} shows that the bootstrap
intervals generally reduce the deviation from nominal coverage in
Settings~A--C relative to the debiased Wald interval, although the magnitude of
the correction differs across the three settings. In Setting~D, both bootstrap
procedures reduce the elevated type~I error observed for the debiased Wald
interval at smaller sample sizes.
\endgroup

\begin{figure}[t]
  \centering
  \includegraphics[width=0.76\linewidth]
  {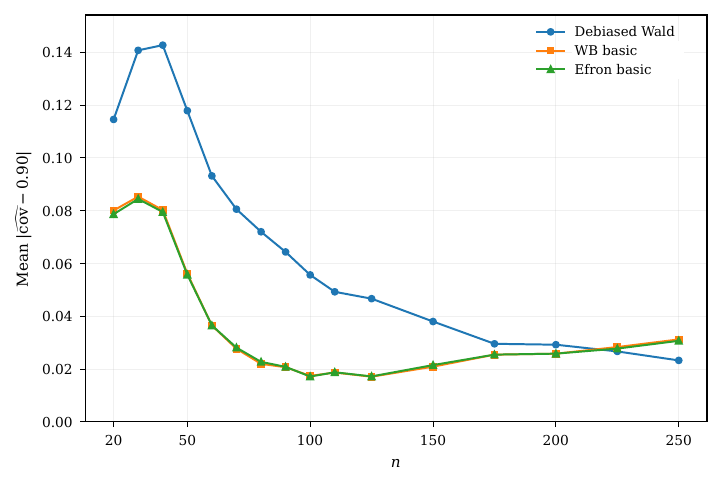}
  \caption{Average absolute coverage error for active coefficients by sample
  size across the settings included in the main comparison using
  \(\lambda_{\min}\) or AIC. Each combination of coefficient and simulation
  setting contributes equally to the average. The basic bootstrap intervals
  have smaller average error than the debiased Wald intervals, with the largest
  practical reduction at smaller and moderate sample sizes.}
  \label{fig:basic-average-error-by-n}
\end{figure}

\begingroup
Figure~\ref{fig:basic-average-error-by-n} summarizes average absolute coverage
error across the settings included in the main comparison by sample size.
The largest reductions relative to the debiased Wald interval occur at smaller
and moderate sample sizes, whereas the differences between the methods decrease
as the sample size increases.
\endgroup

\begin{table}[t]
\centering
\caption{Numerical results for focus Settings~A--D at \(n=100\) and
\(n=250\). For Settings~A--C, the entries report empirical coverage and the
reduction in absolute coverage error relative to the nominal level of 0.90.
For Setting~D, the reported metric is type~I error and the corresponding
reduction is calculated relative to 0.10. ``WB red.'' and ``Efron red.''
denote the reduction relative to the debiased Wald interval; positive values
indicate improvement. Settings~A--D correspond directly to panels
(a)--(d) of Figure~\ref{fig:basic-representative-error-settings}.}
\label{tab:basic-representative-settings}
\scriptsize
\input{basic_interval_outputs_v5/table_representative_settings_v5_latex.tex}
\end{table}

\begingroup
Table~\ref{tab:basic-representative-settings} provides the numerical results
for Settings~A--D displayed in
Figure~\ref{fig:basic-representative-error-settings}. At \(n=100\), the
bootstrap intervals reduce absolute coverage error in each of the three active
Settings~A--C and reduce type~I error in the null Setting~D. At \(n=250\), the
differences are smaller and the bootstrap intervals can become conservative,
illustrating that the finite-sample benefit depends on both sample size and
setting.
\endgroup

\begin{table}[t]
\centering
\caption{Average absolute coverage error for active coefficients by sample
size across the settings included in the main comparison using
\(\lambda_{\min}\) or AIC. To keep the manuscript compact, only the three
interval-error columns are shown.}
\label{tab:basic-average-error-by-n}
\small
\input{basic_interval_outputs_v5/table_average_absolute_error_by_n_v5_latex.tex}
\end{table}

\begingroup
Table~\ref{tab:basic-average-error-by-n} provides the numerical values for the
average absolute coverage errors across the settings included in the main
comparison displayed in Figure~\ref{fig:basic-average-error-by-n}.

The results for interval width illustrate the expected trade-off between
coverage and precision. The basic interval based on the wild bootstrap
provides a moderate correction relative to the debiased Wald interval.
The basic interval based on Efron's bootstrap tends to be more conservative
and can therefore achieve larger coverage gains at the cost of wider
intervals.
\endgroup

\subsubsection*{Computational cost}

\begingroup
The runtime comparison is intended as an indicative computational comparison
rather than as a formal benchmarking study. The wild bootstrap adds moderate
overhead, whereas Efron's bootstrap is substantially more expensive. The
runtime experiment used few repetitions and a small bootstrap size and should
therefore be interpreted with corresponding caution.
Figure~\ref{fig:runtimes} reports the runtime comparison for the investigated
procedures.
\endgroup

\begin{figure}[h]
    \centering
    \includegraphics[width=.8\linewidth]
    {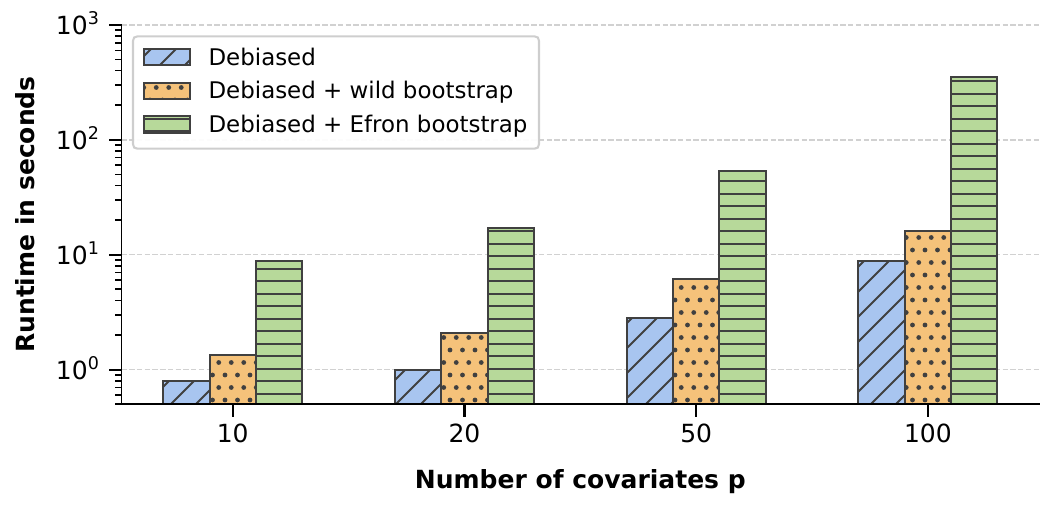}
    \caption{Runtime comparison.}
    \label{fig:runtimes}
\end{figure}

\begingroup
Overall, Tables~\ref{tab:coef-coverage-min} and
\ref{tab:coef-coverage-aic}, which report Reference settings~R1 and~R2,
together with Figure~\ref{fig:basic-representative-error-settings} and
Table~\ref{tab:basic-representative-settings}, which report focus
Settings~A--D, indicate that bootstrap-based intervals can improve finite-sample
coverage for active coefficients after Cox Lasso selection.
Figure~\ref{fig:basic-average-error-by-n} and
Table~\ref{tab:basic-average-error-by-n} show the corresponding pattern in
average absolute coverage error across sample sizes for the settings included
in the main comparison. The wild bootstrap provides a more moderate correction
at lower computational cost, whereas intervals based on Efron's bootstrap tend
to be more conservative. Figure~\ref{fig:runtimes} additionally shows the
substantially higher computational cost of Efron's bootstrap.
\endgroup

\endgroup
\endgroup
\section{Real-data example}
\label{sec:example}

\begingroup
This section illustrates the proposed bootstrap procedures in a real survival
analysis with a moderately large set of baseline covariates using the SUPPORT2
data (Study to Understand Prognoses and Preferences for Outcomes and Risks of
Treatments) \cite{support1995,support2uci}, a multicenter cohort of seriously
ill hospitalized adults in the United States.

We considered the publicly available SUPPORT2 cohort \cite{support2uci}
and analyzed time to death within 180 days after study entry.
Patients who were alive beyond 180 days were treated as right-censored.
The analysis was based on \(n=9{,}105\) patients. Although this sample size is
substantially larger than those considered in the simulation study, the
SUPPORT2 data were chosen as an empirical illustration rather than to reproduce
a particular simulation setting. In particular, the comparatively rich set of
available baseline variables provides a useful setting in which to illustrate
Cox Lasso variable selection followed by the proposed inference procedures.
Baseline information included demographic characteristics, comorbidity
measures, disease-group indicators, physiological variables, and laboratory
measurements. Continuous variables were standardized prior to model fitting,
and categorical variables were represented by dummy variables. For disease
status, we used the \texttt{dzgroup} encoding.
Table~\ref{tab:baseline} summarizes selected baseline characteristics of the
cohort.

\begin{table}[t]
\centering
\caption{Baseline characteristics of the SUPPORT2 cohort. Continuous
characteristics are summarized as mean (SD), except for follow-up, which is
reported as the median. Binary characteristics are reported as percentages,
and disease groups as \(n\) (\%).}
\label{tab:baseline}
\begin{tabular}{l r}
\toprule
\textbf{Characteristic} & \textbf{Summary} \\
\midrule
Age, years & 62.7 (15.6) \\
Female sex, \% & 43.7 \\
Mean blood pressure, mmHg & 84.5 (27.7) \\
Heart rate, beats/min & 97.2 (31.6) \\
Respiratory rate, breaths/min & 23.3 (9.6) \\
Creatinine, mg/dL & 1.8 (1.7) \\
White blood cell count, \(10^3/\mathrm{\mu L}\) & 12.3 (9.2) \\
Sodium, mEq/L & 137.6 (6.0) \\
Number of comorbidities & 1.9 (1.3) \\
\addlinespace
\multicolumn{2}{l}{\textit{Disease group, \(n\) (\%)}} \\
\quad ARF/MOSF with sepsis & 3515 (38.6) \\
\quad CHF & 1387 (15.2) \\
\quad COPD & 967 (10.6) \\
\quad Lung cancer & 908 (10.0) \\
\quad MOSF with malignancy & 712 (7.8) \\
\quad Coma & 596 (6.5) \\
\quad Colon cancer & 512 (5.6) \\
\quad Cirrhosis & 508 (5.6) \\
\addlinespace
Censored observations, \% & 53.2 \\
Median follow-up, days & 180 \\
\bottomrule
\end{tabular}
\end{table}

\begingroup
The analysis aimed to identify variables associated with the hazard of death
within 180 days and to compare uncertainty quantification across the considered
inferential procedures after Cox Lasso variable selection
\cite{tibshirani1997lasso}. Inference was carried out using Wald intervals
based on the debiased Cox estimator \cite{yu2021coxci,kong2021robustcox},
together with intervals based on the wild bootstrap and Efron's bootstrap
\cite{efron1993bootstrap,lin1993,spiekerman1998}.

The Cox Lasso estimator selected 20 covariates.
Figure~\ref{fig:example} displays the estimated hazard ratios for the selected
covariates together with 90\% confidence intervals obtained using the debiased
Wald approximation, the wild bootstrap, and Efron's bootstrap. The bootstrap
procedures use the score-based constructions defined in
Sections~\ref{sec:wild} and~\ref{sec:Efron}; consequently, the original Cox
Lasso fit and the debiasing matrix are kept fixed across bootstrap
replications.

\begin{figure}[htbp]
    \centering
    \includegraphics[width=.88\linewidth]{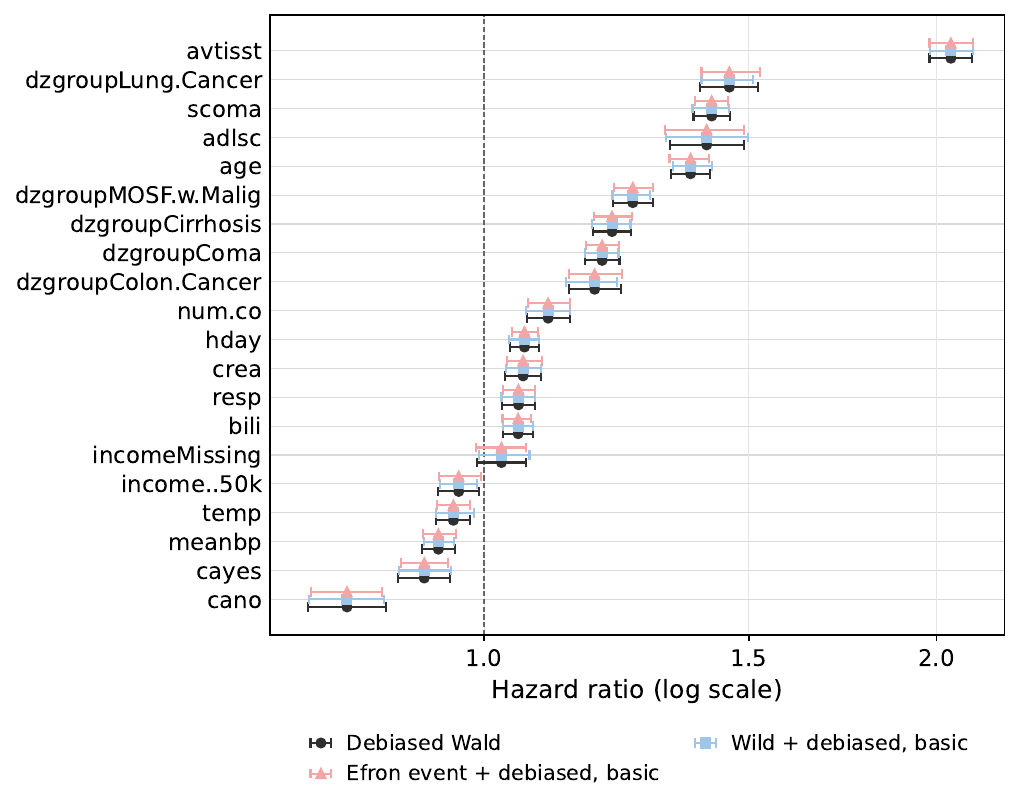}
    \caption{Estimated hazard ratios with 90\% confidence intervals after Cox
    Lasso selection in the SUPPORT2 data. The common hazard-ratio estimates
    are shown with confidence intervals based on the debiased Wald
    approximation, the wild bootstrap, and Efron's bootstrap. }
    \label{fig:example}
\end{figure}

Several selected covariates showed a clear positive association with the hazard
of death. The strongest effects were observed for the severity-related variable
\texttt{avtisst}, for age, and for the coma score \texttt{scoma}.
Among disease-group indicators, elevated hazard ratios were found in particular
for \texttt{dzgroupLung.Cancer}, \texttt{dzgroupComa},
\texttt{dzgroupCirrhosis}, and \texttt{dzgroupMOSF.w.Malig}.
These findings are clinically plausible, as they correspond to severe
underlying disease states and a higher level of physiological instability
within the SUPPORT cohort \cite{support1995,support2uci}.

Several laboratory and physiological variables showed more moderate but still
consistent associations. Higher respiratory rate, bilirubin, creatinine, and
the number of comorbidities were associated with an increased hazard, whereas
higher mean blood pressure and temperature were associated with slightly lower
hazard ratios. The variable \texttt{cano} was estimated below one and appears
protective in the present coding; however, its substantive interpretation
depends on the original SUPPORT variable coding and should therefore be made
with care \cite{support2uci}. Numerical hazard-ratio estimates and 90\%
confidence intervals for all selected covariates are reported in the online
supplement, Table~S12.

Figure~\ref{fig:example} also allows the uncertainty estimates to be compared
across methods. For the strongest signals, the confidence intervals obtained
from the different procedures were broadly similar. For coefficients closer to
the null value, differences in interval width and asymmetry were more noticeable.
In particular, the bootstrap-based procedures sometimes produced wider or more
asymmetric intervals than the debiased Wald interval, reflecting differences
between the finite-sample bootstrap distributions and the first-order
asymptotic approximation
\cite{yu2021coxci,kong2021robustcox,efron1993bootstrap}.

Overall, the SUPPORT2 analysis illustrates how the choice of inferential
procedure can affect uncertainty quantification after Cox Lasso variable
selection. The differences are most visible for more moderate effects, whereas
the conclusions for the strongest associations are comparatively stable across
the considered confidence-interval constructions.
\endgroup
\endgroup

\section{Discussion}
\begingroup
\label{sec:discussion}

\begingroup
In this manuscript, we investigated bootstrap-based inference for the debiased
Cox estimator after Cox Lasso variable selection. We considered both the wild
bootstrap and Efron's bootstrap and compared the resulting confidence intervals
with naive post-selection Wald inference and Wald inference based on the
debiased Cox estimator.

Overall, the simulation study and the real-data example indicate that bootstrap
refinement can improve finite-sample uncertainty quantification, particularly
for moderate effects for which the first-order normal approximation may still be
inaccurate. The SUPPORT2 analysis complements rather than mirrors the simulation
study: the simulations were designed to assess finite-sample behaviour under
small-to-moderate sample sizes, whereas the real-data analysis illustrates the
practical use of the procedures in a substantially larger cohort with a richer
set of candidate predictors. This is consistent with the general motivation for
bootstrap methods as a means of improving finite-sample distributional
approximations
\cite{efron1993bootstrap,davison1997bootstrap,hall1992bootstrap}.

The theoretical results are stated for the Cox Lasso with coordinatewise
\(\ell_1\)-penalization and nodewise \(\ell_1\)-regression in the debiasing
step. This keeps the proof route transparent: once the Cox Lasso estimator is
sufficiently sparse, the nodewise precision estimator is stable, and the
required Cox score linearization is available, the bootstrap argument can be
formulated for the leading score-correction term. In practice, the same
score-based bootstrap construction can be combined with common tuning choices
or closely related weighted or adaptive \(\ell_1\) variants, provided that the
resulting penalized estimator satisfies the sparsity and remainder conditions
required for debiasing \cite{zou2006adaptive,zhang2007adaptive}. Thus, the
assumptions define the setting covered by the present proof rather than
restricting the use of these variants in applied sensitivity analyses.

A further direction is to connect the present bootstrap-based post-selection
inference with focused model-selection ideas. Focused information criteria
choose among candidate models with respect to a pre-specified target parameter
rather than global model fit \cite{claeskens2003fic,claeskens2008model}. Such
target-oriented criteria may be useful when the scientific focus is known in
advance, but we did not pursue this extension here.

There are several directions for further work. One is to develop
a more explicit theory for data-driven tuning rules in debiased
Cox inference, including criteria that select the tuning parameter
with respect to a prespecified inferential target rather than overall
model fit or prediction. Another is
to investigate the classical bootstrap-\(t\) construction with
replicate-specific standard-error estimates and to assess whether this
replicate-specific studentization yields finite-sample improvements that justify the
added computational cost \cite{efron1993bootstrap,davison1997bootstrap}.

A further extension is to move beyond the present componentwise theory toward
simultaneous confidence bands for time-indexed survival quantities, such as the
cumulative baseline hazard or time-varying coefficient processes. Such results
would require functional weak-convergence arguments and could build on wild
bootstrap theory for multiplicative intensity models
\cite{dobler2019confidence,bluhmki2019wild}.

Finally, it would be useful to investigate whether the finite-sample advantages
of bootstrap refinement persist for broader penalty classes, larger covariate
sets, and more complex survival settings such as clustered event times, frailty
structures, or time-dependent covariates.\endgroup

\endgroup

\clearpage
\section*{Supplementary Material}
\addcontentsline{toc}{section}{Supplementary Material}

\begingroup
\renewcommand{\thesection}{S\arabic{section}}
\renewcommand{\thesubsection}{S\arabic{section}.\arabic{subsection}}
\renewcommand{\thesubsubsection}{S\arabic{section}.\arabic{subsection}.\arabic{subsubsection}}
\renewcommand{\thetable}{S\arabic{table}}
\renewcommand{\thefigure}{S\arabic{figure}}
\renewcommand{\theequation}{S\arabic{equation}}
\renewcommand{\thetheorem}{S\arabic{theorem}}
\renewcommand{\thelemma}{S\arabic{lemma}}
\renewcommand{\theassumption}{S\arabic{assumption}}
\renewcommand{\theproposition}{S\arabic{proposition}}
\renewcommand{\thecorollary}{S\arabic{corollary}}
\renewcommand{\theremark}{S\arabic{remark}}
\renewcommand{\theHsection}{supp.\arabic{section}}
\renewcommand{\theHsubsection}{supp.\arabic{section}.\arabic{subsection}}
\renewcommand{\theHsubsubsection}{supp.\arabic{section}.\arabic{subsection}.\arabic{subsubsection}}
\renewcommand{\theHtable}{supp.\arabic{table}}
\renewcommand{\theHfigure}{supp.\arabic{figure}}
\renewcommand{\theHequation}{supp.\arabic{equation}}
\setcounter{section}{0}
\setcounter{subsection}{0}
\setcounter{subsubsection}{0}
\setcounter{table}{0}
\setcounter{figure}{0}
\setcounter{equation}{0}
\setcounter{theorem}{0}
\setcounter{lemma}{0}
\setcounter{assumption}{0}
\setcounter{proposition}{0}
\setcounter{corollary}{0}
\setcounter{remark}{0}

\section{Asymptotic justification for the wild bootstrap}
\label{app:sec:asymptotics_wb}

We justify the asymptotic validity of the wild bootstrap (WB) for inference
based on the debiased Cox estimator after Cox Lasso variable selection in the
classical right-censoring setting.
Inference is conducted componentwise for fixed indices $j\in\widehat M$.

\subsection{Counting-process setup}

This subsection collects notation that is useful for proofs but unnecessarily detailed for the main Methods section. For a vector $\mathbf{x}\in\mathbb{R}^p$, $x^{(j)}$ denotes its $j$th component, $(\cdot)^\top$ denotes transposition, and $\|\mathbf{x}\|_q$ denotes the usual $\ell_q$-norm. For $k=0,1,2$, define
\[
\mathbf{x}^{\otimes 0}=1,
\qquad
\mathbf{x}^{\otimes 1}=\mathbf{x},
\qquad
\mathbf{x}^{\otimes 2}=\mathbf{x}\mathbf{x}^\top.
\]
For theoretical statements, let $M_0=\{j:\beta_j^0\neq0\}$ denote the true active set and $s_0=|M_0|=\|\boldsymbol{\beta}^0\|_0$ its size. The Cox Lasso tuning parameter is assumed to satisfy the usual sparse-rate requirements; a typical order is $\gamma_n\asymp\sqrt{\log(p)/n}$.

Let $(\tilde T_i,\delta_i,\mathbf{X}_i)$, $i=1,\dots,n$, be i.i.d., where
$\tilde T_i=\min(T_i,C_i)$ and $\delta_i=\mathbbm 1(T_i\le C_i)$.
Define the at-risk and counting processes
\begin{align*}
Y_i(t)&=\mathbbm 1(\tilde T_i\ge t),\\
N_i(t)&=\mathbbm 1(\tilde T_i\le t,\ \delta_i=1),
\qquad t\in[0,\tau].
\end{align*}
Let $\{\mathcal F_t\}$ be the usual filtration generated by the observed history.
Under the Cox model with true parameter $\boldsymbol{\beta}^0$ and baseline hazard $\lambda_0$,
the intensity of $N_i$ satisfies
\[
\lambda_i(t \mid \mathcal F_{t^-})
=
Y_i(t)\exp(\mathbf{X}_i^\top\boldsymbol{\beta}^0)\lambda_0(t),
\]
and the associated martingale is
\[
M_i(t)=N_i(t)-\int_0^t Y_i(u)\exp(\mathbf{X}_i^\top\boldsymbol{\beta}^0)\lambda_0(u)\,du.
\]
(See, e.g., \cite{andersen1993model}.)

For $k=0,1,2$, write the empirical risk-set sums
\begin{align*}
S_n^{(k)}(t,\boldsymbol{\beta})
&=
\frac1n\sum_{i=1}^n Y_i(t)\exp(\mathbf{X}_i^\top\boldsymbol{\beta})\mathbf{X}_i^{\otimes k},\\
\bar{\mathbf{X}}_n(t,\boldsymbol{\beta})
&=
\frac{S_n^{(1)}(t,\boldsymbol{\beta})}{S_n^{(0)}(t,\boldsymbol{\beta})}.
\end{align*}
Define the individual score contribution and the normalized aggregate score by
\begin{align*}
\boldsymbol{\psi}_i(\boldsymbol{\beta})
&=
\int_0^\tau \{\mathbf{X}_i-\bar{\mathbf{X}}_n(t,\boldsymbol{\beta})\}\,dN_i(t),\\
\boldsymbol{\dot\ell}_n(\boldsymbol{\beta})
&=
\frac1n\sum_{i=1}^n \boldsymbol{\psi}_i(\boldsymbol{\beta}).
\end{align*}
The unnormalized score is $\mathbf{U}(\boldsymbol{\beta})=\sum_{i=1}^n\boldsymbol{\psi}_i(\boldsymbol{\beta})$.

Equivalently, with
\[
\boldsymbol{\phi}_i(t,\boldsymbol{\beta})=\mathbf{X}_i-\bar{\mathbf{X}}_n(t,\boldsymbol{\beta}),
\]
we may write $\boldsymbol{\psi}_i(\boldsymbol{\beta})=\int_0^\tau \boldsymbol{\phi}_i(t,\boldsymbol{\beta})\,dN_i(t)$. At the true parameter, the Doob--Meyer decomposition yields
\[
\mathbf{U}(\boldsymbol{\beta}^0)=\sum_{i=1}^n\int_0^\tau \boldsymbol{\phi}_i(t,\boldsymbol{\beta}^0)\,dM_i(t),
\]
which is the martingale score representation used to motivate multiplier resampling in the main manuscript.

\subsubsection*{Bootstrap notation}

Let $\mathcal D_n=\{(\tilde T_i,\delta_i,\mathbf{X}_i):i=1,\dots,n\}$
denote the observed data.
For any bootstrap random variable $Z^*$, define the conditional expectation and variance

\begin{align*}
\E^*(Z^*)
&:=
\E(Z^* \mid \mathcal D_n),\\
\Var^*(Z^*)
&:=
\Var(Z^* \mid \mathcal D_n).
\end{align*}

We write $Z_n=o_p^*(1)$ if for every $\varepsilon>0$,
\[
P^*(|Z_n|>\varepsilon)
=
P(|Z_n|>\varepsilon \mid \mathcal D_n)
\xrightarrow{p} 0.
\]

All starred probability statements refer to randomness induced solely by the
bootstrap weights, conditional on the data.

\begingroup
\subsubsection*{Information matrix and nodewise Lasso}

We write
\[
\ell_n(\boldsymbol{\beta})=n^{-1}\ell(\boldsymbol{\beta}),
\qquad
\boldsymbol{\dot\ell}_n(\boldsymbol{\beta})
=
\nabla_{\boldsymbol{\beta}}\ell_n(\boldsymbol{\beta})
=
n^{-1}\mathbf{U}(\boldsymbol{\beta}).
\]
Let \(\Sigma_0\) denote the population information matrix appearing in the
linear expansion of the Cox score and let
\[
\Theta_0=\Sigma_0^{-1}.
\]

To estimate \(\Theta_0\), we use a nodewise Lasso construction as commonly
employed in debiased inference
\cite{vandegeer2014optimal,kong2021robustcox,yu2021coxci,xia2023coxdiverging}.
Let \(\hat{\boldsymbol{\beta}}\) denote the Cox Lasso estimator obtained from
the original data and define
\[
\mathbf{G}_i
=
\boldsymbol{\psi}_i(\hat{\boldsymbol{\beta}})
\in\mathbb{R}^p,
\qquad
\mathbf{G}
=
\begin{pmatrix}
\mathbf{G}_1^\top\\
\vdots\\
\mathbf{G}_n^\top
\end{pmatrix}
\in\mathbb{R}^{n\times p}.
\]
Thus, the rows of \(\mathbf{G}\) contain the individual score contributions
evaluated at the Cox Lasso estimator. For the nodewise regressions, we use
\[
\hat\Sigma
=
\frac{1}{n}\mathbf{G}^\top\mathbf{G}
=
\frac{1}{n}\sum_{i=1}^n
\mathbf{G}_i\mathbf{G}_i^\top.
\]
Under the regularity conditions for the Cox model, this empirical second-moment
matrix has the same population limit as the information matrix appearing in
the asymptotic linear representation
\cite{andersen1993model}.

For each coordinate \(j=1,\ldots,p\), let \(\mathbf{G}_j\) denote the
\(j\)th column of \(\mathbf{G}\) and let \(\mathbf{G}_{-j}\) contain the
remaining columns. We solve the nodewise Lasso problem
\[
\hat{\boldsymbol{\gamma}}_j
=
\arg\min_{\boldsymbol{\gamma}\in\mathbb{R}^{p-1}}
\left\{
\frac{1}{n}
\|\mathbf{G}_j-\mathbf{G}_{-j}\boldsymbol{\gamma}\|_2^2
+
\lambda_j\|\boldsymbol{\gamma}\|_1
\right\},
\]
where \(\lambda_j>0\) is the nodewise tuning parameter. The corresponding
residual variance estimator is
\[
\hat\tau_j^2
=
\frac{1}{n}
\|\mathbf{G}_j-\mathbf{G}_{-j}
\hat{\boldsymbol{\gamma}}_j\|_2^2.
\]
Let \(\hat{\mathbf{c}}_j\in\mathbb{R}^p\) have \(j\)th component equal to
one and the remaining components equal to
\(-\hat{\boldsymbol{\gamma}}_j\). The \(j\)th row of the estimator of
\(\Theta_0\) is then
\[
\hat\Theta_{j\cdot}
=
\frac{\hat{\mathbf{c}}_j^\top}{\hat\tau_j^2}.
\]

For the bootstrap results below, the relevant requirement is that
\(\hat\Theta\) approximates \(\Theta_0\) sufficiently well. This condition is
stated formally in Assumption~(A6).
\subsection{Assumptions}

The following conditions are sufficient for the pointwise (fixed $j$) WB results
stated below. They adapt standard Cox/martingale
regularity conditions to the sparse penalized setting combined with nodewise Lasso requirements.

\begin{enumerate}
\item[(A1)] (Sampling and independent censoring)
The observations $(\tilde T_i,\delta_i,\mathbf{X}_i)$ are i.i.d., and $C_i \perp T_i \mid \mathbf{X}_i$.

\item[(A2)] (Bounded covariates and finite horizon)
There exists $K<\infty$ such that $\|\mathbf{X}_i\|_\infty\le K$ a.s., and $\tau<\infty$.

\item[(A3)] (Baseline hazard and nondegenerate risk sets)
The baseline hazard function $\lambda_0:[0,\tau]\to[0,\infty)$
is uniformly bounded and satisfies
\[
\int_0^\tau \lambda_0(t)\,dt < \infty.
\]
Let
\[
S^{(0)}(t,\boldsymbol{\beta})
=
\E\!\left[ Y_1(t)\exp(\mathbf{X}_1^\top\boldsymbol{\beta}) \right].
\]
Assume that there exists a constant $c_0>0$ such that
\[
\inf_{t\in[0,\tau]} S^{(0)}(t,\boldsymbol{\beta}^0) \ge c_0,
\]
and that
\[
\sup_{t\in[0,\tau]}
\big|
S_n^{(0)}(t,\boldsymbol{\beta}^0)-S^{(0)}(t,\boldsymbol{\beta}^0)
\big|
\xrightarrow{p} 0.
\]

\item[(A4)] (Information matrix and invertibility)
Let
\[
\Sigma_0
:=
-\mathbb{E}
\left[
\nabla^2 \ell_n(\boldsymbol{\beta}^0)
\right]
\]
denote the population information matrix, i.e.\ the negative expectation of the
Hessian of the Cox partial log-likelihood evaluated at the true parameter
$\boldsymbol{\beta}^0$. Assume that $\Sigma_0$ is positive definite and that there exist
constants $0<c_{\min}\le c_{\max}<\infty$ such that
\[
c_{\min}
\le
\lambda_{\min}(\Sigma_0)
\le
\lambda_{\max}(\Sigma_0)
\le
c_{\max}.
\]

\item[(A5)] (Sparsity)
Let $s_0=\|\boldsymbol{\beta}^0\|_0$ and assume
\[
s_0 \frac{\log p}{\sqrt n}\to 0.
\]

\item[(A6)] (Nodewise precision estimation)
The nodewise estimator $\hat\Theta$ satisfies
\[
\|\hat\Theta-\Theta_0\|_\infty=o_p(1)
\]

\item[(A7)] (Moment condition for influence components)
For each fixed $j$,
\[
\mathbb E(\varphi_{1j}^4)<\infty,
\qquad
\varphi_{ij} := \mathbf{e}_j^\top\Theta_0\boldsymbol{\psi}_i(\boldsymbol{\beta}^0).
\]

\item[(A8)] (Wild bootstrap multipliers)
Let $\xi_1,\dots,\xi_n$ be i.i.d., independent of the data, with
$\mathbb E(\xi_i)=0$, $\mathbb E(\xi_i^2)=1$, and $\mathbb E(\xi_i^4)<\infty$.
\end{enumerate}

\begingroup
\subsection{Proofs}

\begin{remark}[What the present proof covers]
The present proof is written for the Cox model with coordinatewise $\ell_1$-penalization together with nodewise $\ell_1$-regression in the debiasing step.
Accordingly, the current argument directly covers the Cox Lasso and closely related weighted or adaptive $\ell_1$ variants only insofar as the same asymptotic linear expansion and remainder control remain valid; see, for example, ~\cite{zou2006adaptive} and ~\cite{zhang2007adaptive} for adaptive $\ell_1$ methodology.

By contrast, the proof does not automatically extend to penalties such as the elastic net~\cite{zou2005elasticnet}, group Lasso~\cite{yuan2006model}, fused Lasso~\cite{tibshirani2005sparsity}, SCAD~\cite{fan2001variable}, or MCP~\cite{zhang2010nearly}.
The reason is that these penalties alter either the KKT structure, the bias term, the target sparsity structure, or the precision-estimation step in a way that is not handled by the present derivation.
Thus, any such extension would require additional arguments and is outside the scope of the current manuscript.
\end{remark}

\subsubsection{Asymptotic linearity of the debiased Cox estimator}

Let $\tilde{\boldsymbol{\beta}}$ denote the debiased estimator (constructed from the Cox Lasso estimator and nodewise precision estimation). The key input is a
linear representation driven by martingale/score contributions.

\begin{theorem}[Asymptotic linearity]\label{app:thm:asymptotic_linearity}
Assume (A1)--(A7). For each fixed $j\in\widehat M$,
\[
\sqrt n(\tilde\beta_j-\beta_j^0)
=
\frac{1}{\sqrt n}\sum_{i=1}^n \varphi_{ij}
+o_p(1).
\]
Consequently,
\[
\sqrt n(\tilde\beta_j-\beta_j^0)\ \Rightarrow\ \mathcal N(0,\sigma_j^2),
\qquad
\sigma_j^2=\Var(\varphi_{1j}).
\]
\end{theorem}

\begin{proof}
This is a standard desparsification argument: expand the (partial) score around
$\boldsymbol{\beta}^0$, plug in the Cox Lasso estimator, and premultiply by an approximate inverse
$\hat\Theta$, so that the leading term becomes the empirical score at $\boldsymbol{\beta}^0$,
while the remainder is controlled by sparsity and nodewise Lasso rates.
For the Cox model under right censoring, the score admits a martingale
representation in the counting-process framework \cite{andersen1993model}.
Cox desparsification results can be invoked as black boxes;
see, e.g., \cite{vandegeer2014optimal} for the general blueprint and
\cite{yu2021coxci,xia2023coxdiverging} for Cox-specific statements.
\end{proof}

\begin{lemma}[Conditional Slutsky reduction]
\label{app:lem:cond-slutsky}
Let \(Z_n^*\) and \(R_n^*\) be bootstrap random variables and let
\[
P^*(\cdot)=P(\cdot\mid\mathcal D_n)
\]
denote conditional probability given the data. Suppose that
\[
R_n^*=o_p^*(1)
\]
and that, conditionally on the data, \(Z_n^*\) converges weakly in probability
to a random variable \(Z\) with a continuous distribution function. Then
\[
\sup_{t\in\mathbb R}
\left|
P^*(Z_n^*+R_n^*\le t)
-
P^*(Z_n^*\le t)
\right|
\xrightarrow{p}0.
\]
\end{lemma}

\begin{proof}
Fix $\varepsilon>0$.
By the elementary set inclusions
\[
\{Z_n^*\le t-\varepsilon\}\cap\{|R_n^*|\le \varepsilon\}
\subseteq
\{Z_n^*+R_n^*\le t\}
\subseteq
\{Z_n^*\le t+\varepsilon\}\cup\{|R_n^*|>\varepsilon\},
\]
we obtain
\[
P^*(Z_n^*\le t-\varepsilon)-P^*(|R_n^*|>\varepsilon)
\le
P^*(Z_n^*+R_n^*\le t)
\]
and
\[
P^*(Z_n^*+R_n^*\le t)
\le
P^*(Z_n^*\le t+\varepsilon)+P^*(|R_n^*|>\varepsilon).
\]
Hence
\[
\sup_{t\in\mathbb R}
\left|
P^*(Z_n^*+R_n^*\le t)-P^*(Z_n^*\le t)
\right|
\]
is bounded by
\[
\sup_{t\in\mathbb R}
\left|
P^*(Z_n^*\le t+\varepsilon)-P^*(Z_n^*\le t-\varepsilon)
\right|
+P^*(|R_n^*|>\varepsilon).
\]
Since \(R_n^*=o_p^*(1)\), the second term converges to zero in probability.
By the assumed conditional weak convergence of \(Z_n^*\) to a continuous
limit distribution, the first term can be made arbitrarily small by first
choosing \(\varepsilon\) sufficiently small and then letting \(n\to\infty\).
This proves the claim.
\end{proof}

\begin{remark}
\label{app:rem:no_independence_lemma1}
Lemma~\ref{app:lem:cond-slutsky} does not require independence between
$Z_n^*$ and $R_n^*$.
The argument is purely based on set inclusions and conditional probability bounds.
Any independence assumptions enter only at the subsequent step where the conditional distribution of the leading bootstrap term $Z_n^*$ is identified.

In particular, for the wild bootstrap, this step relies on the i.i.d. multiplier structure and independence of the multipliers from the data.
For Efron's bootstrap, the corresponding argument is instead based on exchangeable multinomial weights, which are independent of the data but not mutually independent.
\end{remark}

The role of Lemma~\ref{app:lem:cond-slutsky} is therefore only to show that the bootstrap remainder term is asymptotically negligible in conditional probability.
The actual bootstrap approximation is determined by the leading term.
Accordingly, the additional assumptions used in the proofs of bootstrap consistency concern the structure of that leading term:
for the wild bootstrap, i.i.d. multipliers independent of the data are used,
whereas for Efron's bootstrap the argument is based on centered multinomial bootstrap weights.

\endgroup

\subsubsection{Consistency of the wild bootstrap}

Recall that, conditionally on the observed data, the wild bootstrap replicate
is defined by
\[
\tilde{\boldsymbol{\beta}}^{\mathrm{WB}}
=
\tilde{\boldsymbol{\beta}}
+
\hat\Theta
\frac{1}{n}
\sum_{i=1}^n
\xi_i\boldsymbol{\psi}_i(\hat{\boldsymbol{\beta}}),
\]
where \(\hat{\boldsymbol{\beta}}\) and \(\hat\Theta\) are estimated from the
original data and kept fixed throughout the bootstrap. Hence, for coordinate
\(j\),
\[
\sqrt n
\left(
\tilde\beta_j^{\mathrm{WB}}-\tilde\beta_j
\right)
=
\mathbf e_j^\top\hat\Theta
\frac{1}{\sqrt n}
\sum_{i=1}^n
\xi_i\boldsymbol{\psi}_i(\hat{\boldsymbol{\beta}}).
\]

Under assumptions (A5)--(A8), replacing the estimated quantities in this
expression by their population counterparts contributes only a conditionally
negligible remainder. Therefore,
\[
\sqrt n
\left(
\tilde\beta_j^{\mathrm{WB}}-\tilde\beta_j
\right)
=
\frac{1}{\sqrt n}
\sum_{i=1}^n
\xi_i\varphi_{ij}
+
o_p^*(1),
\]
where
\[
\varphi_{ij}
=
\mathbf e_j^\top\Theta_0
\boldsymbol{\psi}_i(\boldsymbol{\beta}^0).
\]

\begin{theorem}[Wild bootstrap consistency]\label{app:thm:wb-cons}
Assume (A1)--(A8). For each fixed \(j\in\widehat M\),
\begin{align*}
\sup_{t\in\mathbb R}
\bigg|
& P^*\!\left(
\sqrt n(\tilde\beta_j^{\mathrm{WB}}-\tilde\beta_j)\le t
\right) \\
&\quad -
P\!\left(
\sqrt n(\tilde\beta_j-\beta_j^0)\le t
\right)
\bigg|
\xrightarrow{p}0.
\end{align*}
\end{theorem}

To transfer the conditional convergence of the leading bootstrap term to the
full bootstrap statistic, we use
Lemma~\ref{app:lem:cond-slutsky}.

\begin{proof}
Fix \(j\in\widehat M\) and recall
\[
\varphi_{ij}
=
\mathbf e_j^\top\Theta_0
\boldsymbol{\psi}_i(\boldsymbol{\beta}^0)
\]
from (A7). Define the leading terms
\[
U_{n,j}
:=
\frac{1}{\sqrt n}
\sum_{i=1}^n
\varphi_{ij},
\qquad
U_{n,j}^*
:=
\frac{1}{\sqrt n}
\sum_{i=1}^n
\xi_i\varphi_{ij},
\]
and the corresponding distribution functions
\begin{align*}
F_{n,j}(t)
&:=
P\!\left(
\sqrt n(\tilde\beta_j-\beta_j^0)\le t
\right),\\
F_{n,j}^*(t)
&:=
P^*\!\left(
\sqrt n(\tilde\beta_j^{\mathrm{WB}}-\tilde\beta_j)\le t
\right).
\end{align*}

\medskip
\noindent
\textit{Step 1: Reduction to the leading terms.}

By Theorem~\ref{app:thm:asymptotic_linearity} and assumptions (A1)--(A7),
\begin{equation}\label{app:eq:orig-lin}
\sqrt n(\tilde\beta_j-\beta_j^0)
=
U_{n,j}+r_{n,j},
\qquad
r_{n,j}=o_p(1).
\end{equation}
Moreover, the expansion above gives
\begin{equation}\label{app:eq:wb-lin2}
\sqrt n(\tilde\beta_j^{\mathrm{WB}}-\tilde\beta_j)
=
U_{n,j}^*+r_{n,j}^*,
\qquad
r_{n,j}^*=o_p^*(1).
\end{equation}

For the original statistic, ordinary Slutsky arguments together with the
continuity of the limiting distribution imply
\[
\sup_{t\in\mathbb R}
\left|
P\!\left(
\sqrt n(\tilde\beta_j-\beta_j^0)\le t
\right)
-
P(U_{n,j}\le t)
\right|
\to0.
\]
For the bootstrap statistic,
Lemma~\ref{app:lem:cond-slutsky} and
\eqref{app:eq:wb-lin2} yield
\[
\sup_{t\in\mathbb R}
\left|
P^*\!\left(
\sqrt n(\tilde\beta_j^{\mathrm{WB}}-\tilde\beta_j)\le t
\right)
-
P^*(U_{n,j}^*\le t)
\right|
\xrightarrow{p}0.
\]
Consequently, it remains to show
\begin{equation}\label{app:eq:suff}
\sup_{t\in\mathbb R}
\left|
P^*(U_{n,j}^*\le t)-P(U_{n,j}\le t)
\right|
\xrightarrow{p}0.
\end{equation}

\medskip
\noindent
\textit{Step 2: Limiting distribution of \(U_{n,j}\).}

By Theorem~\ref{app:thm:asymptotic_linearity} and
\eqref{app:eq:orig-lin},
\[
U_{n,j}
\Rightarrow
\mathcal N(0,\sigma_j^2),
\qquad
\sigma_j^2=\Var(\varphi_{1j}).
\]
Since the limiting distribution is continuous,
\begin{equation}\label{app:eq:F-limit}
\sup_{t\in\mathbb R}
\left|
P(U_{n,j}\le t)
-
\Phi(t/\sigma_j)
\right|
\to0,
\end{equation}
where \(\Phi\) denotes the standard normal distribution function.

\medskip
\noindent
\textit{Step 3: Conditional multiplier approximation.}

Conditionally on the observed data, the quantities
\(\{\varphi_{ij}\}_{i=1}^n\) are fixed, whereas the multipliers
\(\{\xi_i\}_{i=1}^n\) are i.i.d., independent of the data, and satisfy the
moment conditions in (A8). Hence
\[
\E^*(U_{n,j}^*)=0,
\qquad
\Var^*(U_{n,j}^*)
=
\frac{1}{n}
\sum_{i=1}^n
\varphi_{ij}^2.
\]
Under (A7) and the corresponding law-of-large-numbers condition,
\begin{equation}\label{app:eq:var-cons}
\Var^*(U_{n,j}^*)
=
\frac{1}{n}
\sum_{i=1}^n
\varphi_{ij}^2
\xrightarrow{p}
\sigma_j^2.
\end{equation}
Moreover, the conditional Lindeberg condition holds in probability. In
particular, using the fourth-moment conditions in (A7)--(A8), for every
\(\varepsilon>0\),
\[
\frac{1}{n}
\sum_{i=1}^n
\E^*\!\left[
\xi_i^2\varphi_{ij}^2
\mathbbm 1
\left\{
|\xi_i\varphi_{ij}|>\varepsilon\sqrt n
\right\}
\right]
\xrightarrow{p}0.
\]
A conditional multiplier central limit theorem therefore yields
\cite{pauly2011weighted}
\begin{equation}\label{app:eq:clt-cond}
U_{n,j}^*
\Rightarrow^*
\mathcal N(0,\sigma_j^2)
\qquad\text{in probability}.
\end{equation}
Since the limiting distribution is continuous,
\eqref{app:eq:clt-cond} implies
\begin{equation}\label{app:eq:kolm-cond}
\sup_{t\in\mathbb R}
\left|
P^*(U_{n,j}^*\le t)
-
\Phi(t/\sigma_j)
\right|
\xrightarrow{p}0.
\end{equation}

\medskip
\noindent
\textit{Step 4: Conclusion.}

Combining \eqref{app:eq:kolm-cond} with
\eqref{app:eq:F-limit} and applying the triangle inequality gives
\begin{align*}
\sup_{t\in\mathbb R}
\left|
P^*(U_{n,j}^*\le t)-P(U_{n,j}\le t)
\right|
&\le
\sup_t
\left|
P^*(U_{n,j}^*\le t)
-
\Phi(t/\sigma_j)
\right|\\
&\quad+
\sup_t
\left|
P(U_{n,j}\le t)
-
\Phi(t/\sigma_j)
\right|
\xrightarrow{p}0.
\end{align*}
This proves \eqref{app:eq:suff}. Together with the reductions in Step~1,
the asserted bootstrap consistency follows.
\end{proof}
\subsubsection{Consistency of the wild bootstrap variance estimate}

Let $\tilde\beta_j^{\mathrm{WB},(1)},\ldots,\tilde\beta_j^{\mathrm{WB},(B)}$
be WB replicates and define the empirical variance
\begin{align*}
\hat\sigma_{j,\mathrm{WB}}^2
&=
\frac{1}{B-1}\sum_{b=1}^B
(\tilde\beta_{j}^{\mathrm{WB},(b)}-\bar\beta_{j}^{\mathrm{WB}})^2,\\
\bar\beta_{j}^{\mathrm{WB}}
&=
\frac{1}{B}\sum_{b=1}^B \tilde\beta_{j}^{\mathrm{WB},(b)}.
\end{align*}

\begin{theorem}[Wild bootstrap variance consistency]\label{app:thm:wb-var}
Assume (A1)--(A8) and $B\to\infty$. Then, for each fixed $j\in\widehat M$,
\[
\hat\sigma_{j,\mathrm{WB}}^2 \xrightarrow{p} \sigma_j^2/n.
\]
\end{theorem}

\begin{proof}
Conditionally on the data,
\[
\Var^*\!\big(\sqrt n(\tilde\beta_j^{\mathrm{WB}}-\tilde\beta_j)\big)
=
\frac{1}{n}\sum_{i=1}^n \varphi_{ij}^2 + o_p(1),
\]
using $\E(\xi_i^2)=1$ and the negligible remainder.
By the law of large numbers and (A7),
$n^{-1}\sum_{i=1}^n \varphi_{ij}^2 \to_p \E(\varphi_{1j}^2)=\sigma_j^2$.
Hence the conditional variance of $(\tilde\beta_j^{\mathrm{WB}}-\tilde\beta_j)$
converges to $\sigma_j^2/n$. The sample variance over $B$ replicates consistently
estimates the conditional variance as $B\to\infty$; see, e.g.,
\cite{vanderVaartWellner1996,kosorok2008empirical,bluhmki2019wild}.
\end{proof}

\subsubsection{Validity of the standardized wild bootstrap statistic}

Define the standardized statistics
\begin{align*}
T_j
&=
\frac{\sqrt n(\tilde\beta_j-\beta_j^0)}{\sigma_j},
\\[2mm]
T_j^{\mathrm{WB}}
&=
\frac{\sqrt n(\tilde\beta_j^{\mathrm{WB}}-\tilde\beta_j)}{\sqrt n\,\hat\sigma_{j,\mathrm{WB}}}
=
\frac{\tilde\beta_j^{\mathrm{WB}}-\tilde\beta_j}{\hat\sigma_{j,\mathrm{WB}}}.
\end{align*}

\begin{theorem}[Validity of the standardized wild bootstrap statistic]\label{app:thm:stud-wb}
Assume (A1)--(A8) and $B\to\infty$. For each fixed $j\in\widehat M$,
\[
\sup_{t\in\mathbb R}
\left|
P^*\!\left(T_j^{\mathrm{WB}}\le t\right)
-
P\!\left(T_j\le t\right)
\right|
\xrightarrow{p} 0.
\]
\end{theorem}

\begin{proof}
The result follows from Theorems~\ref{app:thm:wb-cons}
and~\ref{app:thm:wb-var}.
Indeed, Theorem~\ref{app:thm:wb-cons} establishes bootstrap consistency
for the unstandardized statistic, while Theorem~\ref{app:thm:wb-var} gives
\[
\frac{\sqrt n\,\hat\sigma_{j,\mathrm{WB}}}{\sigma_j}
\xrightarrow{p}1.
\]
Hence, by the conditional Slutsky argument of
Lemma~\ref{app:lem:cond-slutsky}, replacing $\sigma_j$ by
$\sqrt n\,\hat\sigma_{j,\mathrm{WB}}$ does not affect the limiting
conditional distribution. Therefore,
\[
\sup_{t\in\mathbb R}
\left|
P^*\!\left(T_j^{\mathrm{WB}}\le t\right)
-
P\!\left(T_j\le t\right)
\right|
\xrightarrow{p}0.
\]
\end{proof}

\begin{corollary}[Asymptotic coverage of the basic wild-bootstrap interval]
\label{app:cor:basic-wb}
Assume (A1)--(A8), $B\to\infty$, and continuity of the limiting distribution
of $T_j$. Let $\hat q_{j,\gamma}^{\Delta,\mathrm{WB}}$ denote the empirical
$\gamma$-quantile of the wild bootstrap deviations
$\Delta_j^{\mathrm{WB},(b)}=\tilde\beta_j^{\mathrm{WB},(b)}-\tilde\beta_j$.
Then the basic wild-bootstrap interval
\[
\CI_j^{\mathrm{basic,WB}}(1-\alpha)
=
\Big[
\tilde\beta_j-\hat q_{j,1-\alpha/2}^{\Delta,\mathrm{WB}},
\ \tilde\beta_j-\hat q_{j,\alpha/2}^{\Delta,\mathrm{WB}}
\Big]
\]
satisfies
\[
P\!\left(
\beta_j^0\in\CI_j^{\mathrm{basic,WB}}(1-\alpha)
\right)
\to 1-\alpha.
\]
\end{corollary}

The connection with Theorem~\ref{app:thm:stud-wb} is algebraic. Because
$\hat\sigma_{j,\mathrm{WB}}$ is common to all bootstrap replications,
quantiles of
$T_j^{\mathrm{WB}}=(\tilde\beta_j^{\mathrm{WB}}-\tilde\beta_j)/
\hat\sigma_{j,\mathrm{WB}}$ are the corresponding deviation quantiles divided
by $\hat\sigma_{j,\mathrm{WB}}$. Hence quantile inversion of the standardized
statistic yields exactly the basic interval above.
\endgroup

\begingroup
\section{Asymptotic justification for Efron's bootstrap} \label{app:sec:asymptotics_eb}

Let $(W_1,\dots,W_n)\sim \mathrm{Mult}\!\left(n;\frac1n,\dots,\frac1n\right)$ be
multinomial bootstrap weights, independent of the data, and define
$\bar W_i:=W_i-1$.
Define the analogue for Efron's bootstrap (conditionally on the data) by
the weighted version of the linear term:
\[
\sqrt n(\tilde\beta_j^{\mathrm{Efr}}-\tilde\beta_j)
=
\frac{1}{\sqrt n}\sum_{i=1}^n \bar W_i \,\varphi_{ij}
+o_p^*(1),
\]
where $o_p^*(1)$ denotes convergence to $0$ in conditional probability.
The following results are established under assumptions (A1)--(A7).
For Efron's bootstrap we additionally require the following condition.

\subsubsection*{(A9) Efron's bootstrap weights}
Let
\[
(W_1,\dots,W_n)
\sim
\mathrm{Mult}\!\left(n;\frac1n,\dots,\frac1n\right)
\]
be independent of the data, and set $\bar W_i=W_i-1$.

\subsection{Proofs}

This subsection provides the detailed arguments underlying the asymptotic
validity of Efron's bootstrap. We first establish consistency of the bootstrap
distribution, then prove consistency of the corresponding bootstrap variance
estimator, and finally combine both results to establish validity of the
standardized bootstrap statistic and the resulting basic bootstrap interval.

\subsubsection{Consistency of Efron's bootstrap}

We write the argument for Efron's bootstrap in the same two-layer form
as the proof for the wild bootstrap.
The first layer is Cox-specific: Theorem~\ref{app:thm:asymptotic_linearity}
and the bootstrap linearization reduce the original and bootstrap statistics
to leading sums of influence contributions, up to negligible remainders.
For the second layer, we apply Theorem~4.1 of
Pauly~\cite{pauly2011weighted} to the resulting centered
multinomially weighted sum.
Below, we explicitly identify the statistic for Efron's bootstrap with
Pauly's weighted-resampling statistic and verify the required conditions.

\begin{theorem}[Efron's bootstrap consistency]
\label{app:thm:efr-cons}
Assume (A1)--(A7) and (A9). For each fixed $j\in\widehat M$,
\begin{align*}
\sup_{t\in\mathbb R}
\bigg|
& P^*\!\left(
\sqrt n(\tilde\beta_j^{\mathrm{Efr}}-\tilde\beta_j)\le t
\right) \\
&\quad -
P\!\left(
\sqrt n(\tilde\beta_j-\beta_j^0)\le t
\right)
\bigg|
\xrightarrow{p}0.
\end{align*}
\end{theorem}

\begin{proof}
Fix $j\in\widehat M$ and define
\begin{align*}
\varphi_{ij}
&=\mathbf{e}_j^\top\Theta_0\boldsymbol{\psi}_i(\boldsymbol{\beta}^0),\\
U_{n,j}
&:=\frac{1}{\sqrt n}\sum_{i=1}^n\varphi_{ij},\\
U_{n,j}^{\mathrm{Efr},*}
&:=\frac{1}{\sqrt n}\sum_{i=1}^n(W_i-1)\varphi_{ij},
\end{align*}
where
\[
(W_1,\ldots,W_n)
\sim
\operatorname{Mult}\!\left(
n;\frac1n,\ldots,\frac1n
\right)
\]
is independent of the data.

By Theorem~\ref{app:thm:asymptotic_linearity} and the bootstrap
linearization,
\begin{align*}
\sqrt n(\tilde\beta_j-\beta_j^0)
&=
U_{n,j}+r_{n,j},\\
\sqrt n(\tilde\beta_j^{\mathrm{Efr}}-\tilde\beta_j)
&=
U_{n,j}^{\mathrm{Efr},*}+r_{n,j}^*,
\end{align*}
where
\[
r_{n,j}=o_p(1),
\qquad
r_{n,j}^*=o_p^*(1).
\]
By Lemma~\ref{app:lem:cond-slutsky}, it therefore suffices to
establish the corresponding distributional approximation for the
leading terms.

The unconditional convergence
\[
U_{n,j}\Rightarrow\mathcal N(0,\sigma_j^2)
\]
follows from the ordinary central limit theorem under (A1) and (A7),
as in the proof of the wild bootstrap consistency result.

For the conditional bootstrap approximation, we apply
Theorem~4.1 of Pauly~\cite{pauly2011weighted}.
Set
\[
Z_{n,i}
:=
\frac{\varphi_{ij}}{\sqrt n},
\qquad
W_{n,i}^{P}
:=
\frac{W_i-1}{\sqrt n},
\]
where the superscript $P$ distinguishes the weights in Pauly's
notation from the multinomial counts $W_i$ used here.
For $k(n)=m(n)=n$, the weights $W_{n,i}^{P}$ coincide with the
centered multinomial weights in equation~(2.2) of
Pauly~\cite{pauly2011weighted}.
Moreover, because
\[
\sum_{i=1}^n(W_i-1)=0,
\]
Pauly's weighted-resampling statistic becomes
\begin{align*}
\sqrt n\sum_{i=1}^n
W_{n,i}^{P}(Z_{n,i}-\bar Z_n)
&=
\sum_{i=1}^n(W_i-1)Z_{n,i}\\
&=
\frac{1}{\sqrt n}
\sum_{i=1}^n(W_i-1)\varphi_{ij}\\
&=
U_{n,j}^{\mathrm{Efr},*}.
\end{align*}

It remains to verify the two conditions of
Theorem~4.1 of Pauly~\cite{pauly2011weighted}.
First, by (A7),
\[
\max_{1\le i\le n}|Z_{n,i}|
=
\frac{1}{\sqrt n}
\max_{1\le i\le n}|\varphi_{ij}|
\xrightarrow{p}0.
\]
Indeed, for every $\varepsilon>0$, the finite fourth moment in (A7)
and Markov's inequality imply
\begin{align*}
P\!\left(
\max_{1\le i\le n}|\varphi_{ij}|
>
\varepsilon\sqrt n
\right)
&\le
n\,
P\!\left(
|\varphi_{1j}|>\varepsilon\sqrt n
\right)\\
&\le
\frac{\E(\varphi_{1j}^4)}
{\varepsilon^4 n}
\longrightarrow0.
\end{align*}

Second, writing
\[
\bar\varphi_j
:=
\frac1n\sum_{i=1}^n\varphi_{ij},
\]
we have, by (A1), (A7), and the law of large numbers,
\begin{align*}
\sum_{i=1}^n
(Z_{n,i}-\bar Z_n)^2
&=
\frac1n
\sum_{i=1}^n
(\varphi_{ij}-\bar\varphi_j)^2\\
&\xrightarrow{p}
\Var(\varphi_{1j})
=
\sigma_j^2.
\end{align*}
Thus, the conditions of Theorem~4.1 of
Pauly~\cite{pauly2011weighted} are satisfied, and hence
\[
\sup_{t\in\mathbb R}
\left|
P^*\!\left(
U_{n,j}^{\mathrm{Efr},*}\le t
\right)
-
\Phi(t/\sigma_j)
\right|
\xrightarrow{p}0.
\]
At the same time, the unconditional central limit theorem and
continuity of the Gaussian distribution function imply
\[
\sup_{t\in\mathbb R}
\left|
P(U_{n,j}\le t)
-
\Phi(t/\sigma_j)
\right|
\to0.
\]
Therefore, by the triangle inequality,
\[
\sup_{t\in\mathbb R}
\left|
P^*\!\left(
U_{n,j}^{\mathrm{Efr},*}\le t
\right)
-
P(U_{n,j}\le t)
\right|
\xrightarrow{p}0.
\]
Finally, Lemma~\ref{app:lem:cond-slutsky} transfers this
convergence from the leading terms to the full statistics, which
proves the claim.
\end{proof}
\subsubsection{Consistency of Efron's bootstrap variance estimate}

Let $\tilde\beta_j^{\mathrm{Efr},(1)},\ldots,
\tilde\beta_j^{\mathrm{Efr},(B)}$ be Efron's bootstrap replicates and define
\begin{align*}
\hat\sigma_{j,\mathrm{Efr}}^2
&=
\frac{1}{B-1}\sum_{b=1}^B
\big(\tilde\beta_{j}^{\mathrm{Efr},(b)}
-\bar\beta_{j}^{\mathrm{Efr}}\big)^2,\\
\bar\beta_{j}^{\mathrm{Efr}}
&=
\frac{1}{B}\sum_{b=1}^B
\tilde\beta_{j}^{\mathrm{Efr},(b)}.
\end{align*}

\begin{theorem}[Variance consistency for Efron's bootstrap]
\label{app:thm:efr-var}
Assume (A1)--(A7), (A9), and $B\to\infty$.
Then, for each fixed $j\in\widehat M$,
\[
\hat\sigma_{j,\mathrm{Efr}}^2
\xrightarrow{p}
\sigma_j^2/n.
\]
\end{theorem}

\begin{proof}
The argument follows the proof of Theorem~\ref{app:thm:wb-var}, with one
change: the centered multinomial weights are exchangeable rather than
independent. Conditionally on the data,
\[
\Var^*(\bar W_i)=1-\frac1n,
\qquad
\Cov^*(\bar W_i,\bar W_k)=-\frac1n
\quad (i\neq k),
\]
so that
\[
\Var^*\!\left(
\frac1{\sqrt n}\sum_{i=1}^n\bar W_i\varphi_{ij}
\right)
=
\frac1n\sum_{i=1}^n\varphi_{ij}^2
-
\frac1{n^2}\left(\sum_{i=1}^n\varphi_{ij}\right)^2
\xrightarrow{p}\sigma_j^2
\]
by (A1) and (A7). The remainder control and the Monte-Carlo variance
argument are the same as for the wild bootstrap, which yields
\[
\hat\sigma_{j,\mathrm{Efr}}^2
\xrightarrow{p}
\frac{\sigma_j^2}{n}.
\]
\end{proof}

\subsubsection{Validity of the standardized bootstrap statistic for Efron's bootstrap}

Define
\begin{align*}
T_j
&=
\frac{\sqrt n(\tilde\beta_j-\beta_j^0)}{\sigma_j},
\\[2mm]
T_j^{\mathrm{Efr}}
&=
\frac{\sqrt n(\tilde\beta_j^{\mathrm{Efr}}-\tilde\beta_j)}
{\sqrt n\,\hat\sigma_{j,\mathrm{Efr}}}
=
\frac{\tilde\beta_j^{\mathrm{Efr}}-\tilde\beta_j}
{\hat\sigma_{j,\mathrm{Efr}}}.
\end{align*}

\begin{theorem}[Validity of the standardized bootstrap statistic for Efron's bootstrap]
\label{app:thm:stud-efr}
Assume (A1)--(A7), (A9), and $B\to\infty$.
For each fixed $j\in\widehat M$,
\[
\sup_{t\in\mathbb R}
\left|
P^*\!\left(T_j^{\mathrm{Efr}}\le t\right)
-
P\!\left(T_j\le t\right)
\right|
\xrightarrow{p}0.
\]
\end{theorem}

\begin{proof}
The argument is the same as for the standardized wild bootstrap statistic.
Theorem~\ref{app:thm:efr-cons} establishes consistency of Efron's
bootstrap for the unstandardized statistic, while
Theorem~\ref{app:thm:efr-var} gives
\[
\frac{\sqrt n\,\hat\sigma_{j,\mathrm{Efr}}}{\sigma_j}
\xrightarrow{p}1.
\]
Hence, by the conditional Slutsky argument of
Lemma~\ref{app:lem:cond-slutsky}, replacing $\sigma_j$ by
$\sqrt n\,\hat\sigma_{j,\mathrm{Efr}}$ does not alter the limiting
conditional distribution, and the result follows.
\end{proof}

As for the wild bootstrap, the standard-error estimate
$\hat\sigma_{j,\mathrm{Efr}}$ is common to all bootstrap replications.
Consequently, quantile inversion of the standardized statistic yields the
basic interval based on the deviations from Efron's bootstrap. Under the
conditions above and $B\to\infty$, this interval has asymptotic coverage
$1-\alpha$.
\endgroup

\begingroup
\subsection{Extension to functional convergence in $D[0,\tau]$}
\label{app:remark_D}

The results in Sections~\ref{app:sec:asymptotics_wb}--\ref{app:sec:asymptotics_eb} are formulated
\emph{componentwise} for fixed indices $j\in\widehat M$.
In this setting, it is sufficient to work with scalar central limit theorems
for the debiased estimator and their multiplier (wild bootstrap) analogues.
Concretely, the key step is an asymptotic linear expansion of the form
\[
\sqrt n(\tilde\beta_j-\beta_j^0)
=
\frac{1}{\sqrt n}\sum_{i=1}^n \varphi_{ij}+o_p(1),
\]
together with the corresponding conditional expansion
\[
\sqrt n(\tilde\beta_j^{\mathrm{WB}}-\tilde\beta_j)
=
\frac{1}{\sqrt n}\sum_{i=1}^n \xi_i\varphi_{ij}+o_p^*(1),
\]
which implies weak convergence to a univariate normal distribution and
bootstrap consistency. In particular, no functional limit theory is required
as long as inference targets are finite-dimensional and fixed.

Functional convergence in the Skorokhod space $D[0,\tau]$ becomes relevant when
the inferential object is itself a \emph{stochastic process} indexed by time.
Typical examples include simultaneous confidence bands for the baseline
cumulative hazard, time-indexed score processes, or other functionals that
depend on the entire trajectory $\{N_i(t),Y_i(t):t\in[0,\tau]\}$; see, for
example, \cite{dobler2019confidence,ruehl2024resampling,dietrich2025wild}.
In such cases, one aims to establish weak convergence of the underlying
martingale/estimating process in $D[0,\tau]$, for instance,
\[
\left\{\frac{1}{\sqrt n}\sum_{i=1}^n \int_0^t H_i(u)\,dM_i(u)\right\}_{t\in[0,\tau]}
\Rightarrow \mathbb G(\cdot)
\quad\text{in }D[0,\tau],
\]
for suitable predictable integrands $H_i$ and a mean-zero Gaussian limit process
$\mathbb G$.
Once such a functional CLT holds, convergence can be transferred to suitable
time-indexed functionals, for example to construct simultaneous confidence
bands. Analogously, validity of a time-indexed wild bootstrap requires
functional weak convergence of the corresponding bootstrapped process.
Such wild bootstrap arguments in $D[0,\tau]$ have been developed for
multiplicative hazards models and related counting-process settings
\cite{dobler2019confidence,ruehl2024resampling,dietrich2025wild}.

Technically, establishing convergence in $D[0,\tau]$ typically requires
conditions ensuring predictability of the integrands, suitable moment bounds,
and tightness of the sequence of processes, commonly verified using martingale
weak-convergence arguments. General treatments are given in
\cite{andersen1993model,kosorok2008empirical}, while more specific
wild bootstrap developments for time-to-event and counting-process settings
can be found in
\cite{bluhmki2019wild,dobler2019confidence,ruehl2024resampling,dietrich2025wild}.

For the confidence intervals considered in this manuscript, with fixed coordinates $j$ and finite-dimensional targets, scalar asymptotics suffice. Functional convergence in $D[0,\tau]$ is
only needed if one extends the inferential goals to simultaneous, time-uniform
statements about entire trajectories (e.g., confidence bands).

\endgroup

\begingroup
\section{Simulation details and additional results}
\label{app:sim}

This section collects supplementary numerical material for the simulation study
reported in the main manuscript. The coverage summaries use the full-model
coefficient target described in Section~5 of the main manuscript, while Cox
Lasso selection frequency in the original simulated data set is reported as a
separate operating characteristic. The purpose is not to introduce additional
inferential targets, but to make the finite-sample behaviour behind the
condensed main-text tables more transparent.

Unless stated otherwise, the supplementary simulation summaries are based on
\(M=1000\) Monte Carlo replications and \(B=900\) bootstrap replications, in
accordance with the final result source used for the main simulation results.
Separate computational benchmarks and bootstrap-size sensitivity analyses use
the replication budgets stated explicitly for those analyses.

The section is organized as follows. Subsection~\ref{app:basic-many-settings} reports additional results for individual simulation
settings for the basic wild-bootstrap and Efron-bootstrap intervals. Subsection~\ref{app:coefficient-specific-tables}
reports coefficient-specific coverage for Reference settings~R1 and~R2.
Subsection~\ref{app:selection-frequency-tuning-grid} compares all four tuning
rules under both the standard and adaptive Cox Lasso in the same simulation
setting. Subsection~\ref{app:additional-ci} gives the corresponding summaries
of interval widths for individual coefficients in R1 and R2.

\subsection{Supplementary simulation summaries}
\label{app:additional-sim}

The results below complement the main simulation figures and tables. For active
coefficients, empirical coverage and absolute coverage error relative to the
nominal 90\% level are evaluated over Monte Carlo replications for which the
corresponding interval is available. Inactive coefficients are used as a
type~I-error diagnostic. Cox Lasso selection frequencies from the original
simulated data sets are reported separately.

\subsubsection{Results for individual simulation settings}
\label{app:basic-many-settings}

To complement the representative main-text displays, this subsection reports
additional summaries for selected simulation settings for the basic
bootstrap intervals. The
existing coverage-error displays retain the \(\lambda_{\min}\) and AIC settings
emphasized in the main manuscript. The full simulation design additionally
includes \(\lambda_{1se}\) and BIC, and each of the four tuning rules was
investigated under both the standard and adaptive Cox Lasso specifications.
The summaries should be read descriptively: the main comparison is based on
closeness to the nominal level, whereas the supplementary material additionally
documents the sensitivity to the tuning rule and penalty specification.

\begin{center}
\captionof{table}{Labels and roles of the methods emphasized in the simulation summaries.}
\label{tab:basic-method-labels}
\small
\input{basic_interval_outputs_v5/table_method_labels_v5_latex.tex}
\end{center}

\begin{center}
\captionof{table}{Full version of the summary of average absolute coverage error for active coefficients. This table corresponds to the compact main-text table and additionally reports the reductions relative to the plain debiased Wald interval and the number of combinations of coefficients and simulation settings.}
\label{tab:supp-basic-average-error-by-n-full}
\small
\input{basic_interval_outputs_v5/table_average_absolute_error_by_n_v5_latex_full.tex}
\end{center}

\subsubsection{Coefficient-specific coverage tables}
\label{app:coefficient-specific-tables}

The following tables extend the two named reference analyses used in the main
manuscript across the available sample sizes while holding all other simulation
factors fixed. Reference setting~R1 (standard-\(\lambda_{\min}\)) uses the
\textit{realistic} coefficient pattern, a Weibull baseline hazard, \(p=10\),
10\% target censoring, independent covariates (\(\rho=0\)), and the standard
Cox Lasso with \(\lambda_{\min}\) tuning. Reference setting~R2 (adaptive-AIC)
uses the same data-generating configuration but the adaptive Cox Lasso with AIC
tuning.

Coverage for the debiased and bootstrap procedures is calculated over the
Monte Carlo replications for which the corresponding interval is available.
Cox Lasso selection frequencies are calculated separately from the original
simulated data sets.

\input{generated_supplement_tables_v8/table_coeff_coverage_min_adaptfalse.tex}
\input{generated_supplement_tables_v8/table_coeff_coverage_aic_adapttrue.tex}

\subsubsection{Tuning and penalty-specification sensitivity}
\label{app:selection-frequency-tuning-grid}

To separate the effect of tuning-parameter selection from the effect of
adaptive penalty weighting, we compare all four tuning rules under both
penalty specifications in the same simulation setting. Table~S5 reports
coefficient-specific Cox Lasso selection frequencies from the original
simulated data sets. Tables~S6 and~S7 summarize mean empirical coverage and
mean interval width for the active coefficients for the same eight
combinations. All other simulation factors are held fixed.

\begingroup
\input{generated_supplement_tables_v9/table_selection_frequency_tuning_grid_n80.tex}
\endgroup

\begingroup
\input{generated_supplement_tables_v10/table_tuning_penalty_mean_active_coverage_n80.tex}
\input{generated_supplement_tables_v10/table_tuning_penalty_mean_active_width_n80.tex}
\clearpage
\endgroup

\subsubsection{Interval widths for individual coefficients}
\label{app:additional-ci}

The following tables report mean interval widths for individual coefficients
for the same \(p=10\) series corresponding to Reference setting~R1
(standard-\(\lambda_{\min}\)) and Reference setting~R2 (adaptive-AIC) as used
in the coefficient-specific coverage tables.

\input{generated_supplement_tables_v8/table_coeff_width_min_adaptfalse.tex}
\input{generated_supplement_tables_v8/table_coeff_width_aic_adapttrue.tex}

\endgroup

\section{Additional real-data summary}
\label{sec:supp-support2}

This Supplementary Material collects additional material for the SUPPORT2
real-data example, including a compact summary of the fitted analysis and
numerical results for the covariates selected by the Cox Lasso.

\IfFileExists{support2_numbers.tex}{%
\input{support2_numbers.tex}%
}{%
\newcommand{\SupportN}{9{,}105}%
\newcommand{\SupportDesignP}{44}%
\newcommand{\SupportSelected}{20}%
\newcommand{\SupportEvents}{4{,}265}%
\newcommand{\SupportCensored}{53.2\%}%
\newcommand{\SupportMedianFollowup}{180}%
\newcommand{\SupportEndpoint}{Death within 180 days}%
\newcommand{\SupportLambda}{0.01747}%
}%

\begin{center}
\captionof{table}{Summary of the SUPPORT2 real-data analysis used in the
manuscript.}
\label{tab:supp-support2-summary}
\begin{tabular}{lr}
\toprule
Quantity & Value \\
\midrule
Sample size \(n\) & \SupportN \\
Initial design dimension \(p\) & \SupportDesignP \\
Selected covariates & \SupportSelected \\
Observed events & \SupportEvents \\
Censored observations & \SupportCensored \\
Median follow-up (days) & \SupportMedianFollowup \\
Endpoint & \SupportEndpoint \\
Cox Lasso tuning parameter & \SupportLambda \\
\bottomrule
\end{tabular}
\end{center}

\begin{table}[p]
\centering
\caption{Hazard-ratio estimates and 90\% confidence intervals for the 20
covariates selected by the Cox Lasso in the SUPPORT2 analysis. The hazard-ratio
estimate is based on the debiased Cox estimator and is common to all three
inference procedures. The basic intervals are constructed from the bootstrap deviation quantiles
as defined in the main manuscript. For standardized continuous
covariates, hazard ratios correspond to a one-standard-deviation increase.}
\label{tab:supp-support2-selected}
\scriptsize
\resizebox{\linewidth}{!}{%
\begin{tabular}{lcccc}
\toprule
\textbf{Covariate}
& \textbf{HR}
& \textbf{Wald 90\% CI}
& \textbf{Wild basic 90\% CI}
& \textbf{Efron basic 90\% CI} \\
\midrule
\texttt{cano}
& 0.811 & [0.764, 0.860] & [0.758, 0.866] & [0.760, 0.864] \\
\texttt{cayes}
& 0.912 & [0.877, 0.949] & [0.871, 0.955] & [0.868, 0.959] \\
\texttt{meanbp}
& 0.932 & [0.909, 0.957] & [0.906, 0.960] & [0.904, 0.961] \\
\texttt{temp}
& 0.954 & [0.930, 0.979] & [0.922, 0.988] & [0.925, 0.984] \\
\texttt{income..50k}
& 0.962 & [0.932, 0.993] & [0.929, 0.996] & [0.926, 0.999] \\
\texttt{incomeMissing}
& 1.027 & [0.990, 1.066] & [0.981, 1.075] & [0.981, 1.075] \\
\texttt{bili}
& 1.054 & [1.030, 1.078] & [1.027, 1.081] & [1.026, 1.082] \\
\texttt{resp}
& 1.054 & [1.028, 1.081] & [1.024, 1.086] & [1.024, 1.086] \\
\texttt{crea}
& 1.062 & [1.032, 1.092] & [1.029, 1.095] & [1.029, 1.095] \\
\texttt{hday}
& 1.064 & [1.040, 1.088] & [1.033, 1.096] & [1.035, 1.094] \\
\texttt{num.co}
& 1.103 & [1.067, 1.140] & [1.062, 1.146] & [1.061, 1.147] \\
\texttt{dzgroupColon.Cancer}
& 1.185 & [1.139, 1.233] & [1.129, 1.243] & [1.132, 1.239] \\
\texttt{dzgroupComa}
& 1.198 & [1.167, 1.231] & [1.162, 1.235] & [1.164, 1.234] \\
\texttt{dzgroupCirrhosis}
& 1.217 & [1.181, 1.253] & [1.173, 1.262] & [1.172, 1.262] \\
\texttt{dzgroupMOSF.w.Malig}
& 1.256 & [1.218, 1.295] & [1.214, 1.299] & [1.213, 1.300] \\
\texttt{age}
& 1.372 & [1.332, 1.414] & [1.327, 1.419] & [1.325, 1.421] \\
\texttt{adlsc}
& 1.406 & [1.329, 1.488] & [1.306, 1.514] & [1.312, 1.507] \\
\texttt{scoma}
& 1.417 & [1.378, 1.457] & [1.371, 1.466] & [1.375, 1.460] \\
\texttt{dzgroupLung.Cancer}
& 1.456 & [1.393, 1.522] & [1.388, 1.528] & [1.383, 1.533] \\
\texttt{avtisst}
& 2.045 & [1.979, 2.113] & [1.967, 2.126] & [1.963, 2.130] \\
\bottomrule
\end{tabular}%
}
\end{table}
\clearpage

\endgroup

\section*{Ethics approval and consent to participate}
Not applicable. This study is based on simulated data and publicly available, de-identified data.

\section*{Consent for publication}
Not applicable.

\section*{Availability of data and materials}
The data are publicly available via the SUPPORT2 dataset repository and related documentation
\cite{support1995,support2uci}.
The complete code used for the simulation study and data analysis is publicly available at \url{https://github.com/lena-222/BootstrapPartyLassoCox}.

\section*{Competing interests}
The authors declare no competing interests.

\section*{Funding}
This work was supported by the Deutsche Forschungsgemeinschaft (DFG, German Research Foundation) -- project number 439942859.

\section*{Use of generative AI}
During the preparation of this work, the authors used ChatGPT solely to improve the language and clarity of the manuscript. The authors reviewed and edited all generated content as needed and take full responsibility for the content of the manuscript.

\section*{Authors' contributions}
L.S. conceived the study, implemented the methods, performed the analyses, and wrote the manuscript.
S.F.-W. contributed to the methodological development and provided supervision, review, and critical feedback.
All authors read and approved the final manuscript.

\section*{Acknowledgements}
The authors thank the developers of the software packages used in this study and the maintainers of the SUPPORT2 dataset repository for providing open access to the data.

\bibliographystyle{unsrtnat}
\bibliography{references}
\end{document}

%% file: generated_tables_main/table_coverage_min_adaptfalse_p10_n100.tex
\begin{table}[h]
\centering
\caption{{Coefficient-specific empirical coverage for the realistic pattern under Weibull baseline, $\lambda_{\min}$, $p=10$, $n=100$, 10\% target censoring, independent covariates, and no adaptive weights.}}
\label{tab:coef-coverage-min}
\begingroup
\begin{tabular}{llccc}
\hline
Coefficient & Status & Debiased & Wild basic & Efron basic \\
\hline
$\beta_1$ & active & 0.861 & 0.923 & 0.982 \\
$\beta_2$ & active & 0.860 & 0.929 & 0.981 \\
$\beta_3$ & active & 0.858 & 0.933 & 0.976 \\
$\beta_4$ & active & 0.851 & 0.913 & 0.981 \\
$\beta_5$ & inactive & 0.892 & 0.943 & 0.975 \\
$\beta_6$ & inactive & 0.897 & 0.943 & 0.971 \\
\hline
\end{tabular}
\endgroup
\end{table}

%% file: generated_tables_main/table_coverage_aic_adapttrue_p10_n100.tex
\begin{table}[h]
\centering
\caption{{Coefficient-specific empirical coverage for the realistic pattern under Weibull baseline, AIC tuning, $p=10$, $n=100$, 10\% target censoring, independent covariates, and adaptive weights.}}
\label{tab:coef-coverage-aic}

\begin{tabular}{llccc}
\hline
Coefficient & Status & Debiased & Wild basic & Efron basic \\
\hline
$\beta_1$ & active & 0.866 & 0.923 & 0.952 \\
$\beta_2$ & active & 0.864 & 0.910 & 0.934 \\
$\beta_3$ & active & 0.870 & 0.929 & 0.947 \\
$\beta_4$ & active & 0.845 & 0.904 & 0.946 \\
$\beta_5$ & inactive & 0.877 & 0.925 & 0.943 \\
$\beta_6$ & inactive & 0.897 & 0.934 & 0.940 \\
\hline
\end{tabular}

\end{table}

%% file: basic_interval_outputs_v5/table_representative_settings_v5_latex.tex
\resizebox{\linewidth}{!}{%
\begin{tabular}{llrlllllll}
\toprule
Set. & Coef. & $n$ & Target & Deb. Wald & WB basic & Ef. basic & Oracle & WB red. & Ef. red. \\
\midrule
A & $\beta_{4}$ & 100 & Cov. & 0.781 & 0.872 & 0.873 & 0.901 & 0.091 & 0.092 \\
A & $\beta_{4}$ & 250 & Cov. & 0.849 & 0.911 & 0.910 & 0.891 & 0.040 & 0.041 \\
B & $\beta_{2}$ & 100 & Cov. & 0.821 & 0.888 & 0.892 & 0.897 & 0.067 & 0.071 \\
B & $\beta_{2}$ & 250 & Cov. & 0.881 & 0.930 & 0.926 & 0.919 & -0.011 & -0.007 \\
C & $\beta_{1}$ & 100 & Cov. & 0.836 & 0.906 & 0.900 & 0.874 & 0.059 & 0.064 \\
C & $\beta_{1}$ & 250 & Cov. & 0.877 & 0.943 & 0.938 & 0.894 & -0.020 & -0.014 \\
D & $\beta_{10}$ & 100 & T1 & 0.153 & 0.103 & 0.102 & -- & 0.050 & 0.051 \\
D & $\beta_{10}$ & 250 & T1 & 0.102 & 0.063 & 0.060 & -- & -0.034 & -0.038 \\
\bottomrule
\end{tabular}

}

%% file: basic_interval_outputs_v5/table_average_absolute_error_by_n_v5_latex.tex
\begin{tabular}{rrrr}
\toprule
$n$ & Deb. Wald & WB basic & Ef. basic \\
\midrule
100 & 0.056 & 0.017 & 0.017 \\
150 & 0.038 & 0.021 & 0.021 \\
200 & 0.029 & 0.026 & 0.026 \\
250 & 0.023 & 0.031 & 0.031 \\
\bottomrule
\end{tabular}

%% file: basic_interval_outputs_v5/table_method_labels_v5_latex.tex
\resizebox{\linewidth}{!}{%
\begin{tabular}{lll}
\toprule
Label & Interval & Role \\
\midrule
Oracle & Cox model on the true active set & Unattainable benchmark \\
Debiased Wald & Normal interval for the debiased Cox--Lasso estimator & First-order reference \\
WB basic & Basic wild-bootstrap interval after debiasing & Main wild-bootstrap refinement \\
Efron basic & Basic Efron-bootstrap interval after debiasing & Main Efron-bootstrap refinement \\
\bottomrule
\end{tabular}

}

%% file: basic_interval_outputs_v5/table_average_absolute_error_by_n_v5_latex_full.tex
\begin{tabular}{rlllllr}
\toprule
$n$ & Deb. Wald err. & WB basic err. & Ef. basic err. & WB red. & Ef. red. & Settings \\
\midrule
100 & 0.056 & 0.017 & 0.017 & 0.038 & 0.039 & 92 \\
150 & 0.038 & 0.021 & 0.021 & 0.017 & 0.017 & 92 \\
200 & 0.029 & 0.026 & 0.026 & 0.003 & 0.003 & 80 \\
250 & 0.023 & 0.031 & 0.031 & -0.008 & -0.007 & 88 \\
\bottomrule
\end{tabular}

%% file: generated_supplement_tables_v8/table_coeff_coverage_min_adaptfalse.tex
\begingroup
\scriptsize
\setlength{\tabcolsep}{3pt}
\begin{longtable}{rrlrrrrrr}
\caption{Coefficient-specific empirical coverage for Reference setting~R1. The table reports Cox Lasso selection counts and frequencies from the original simulated data sets separately from interval coverage. $N_{\mathrm{CI}}$ is the number of Monte Carlo replications for which all three debiased intervals were available.}\label{tab:supp-coeff-coverage-min}\\
\toprule
$n$ & Coef. & Status & $N_{\mathrm{sel}}$ & Sel. freq. & $N_{\mathrm{CI}}$ & Deb. Wald & WB basic & Ef. basic \\
\midrule
\endfirsthead
\caption[]{Coefficient-specific empirical coverage for Reference setting~R1. (continued)}\\
\toprule
$n$ & Coef. & Status & $N_{\mathrm{sel}}$ & Sel. freq. & $N_{\mathrm{CI}}$ & Deb. Wald & WB basic & Ef. basic \\
\midrule
\endhead
\midrule
\multicolumn{9}{r}{Continued on next page}\\
\endfoot
\bottomrule
\endlastfoot
30 & $\beta_{1}$ & active & 865 & 0.865 & 931 & 0.747 & 0.820 & 0.953 \\
30 & $\beta_{2}$ & active & 832 & 0.832 & 931 & 0.750 & 0.822 & 0.958 \\
30 & $\beta_{3}$ & active & 685 & 0.685 & 931 & 0.778 & 0.845 & 0.926 \\
30 & $\beta_{4}$ & active & 861 & 0.861 & 931 & 0.738 & 0.808 & 0.958 \\
30 & $\beta_{5}$ & inactive & 299 & 0.299 & 931 & 0.861 & 0.912 & 0.987 \\
30 & $\beta_{6}$ & inactive & 309 & 0.309 & 931 & 0.857 & 0.926 & 0.976 \\
40 & $\beta_{1}$ & active & 970 & 0.970 & 984 & 0.809 & 0.872 & 0.970 \\
40 & $\beta_{2}$ & active & 953 & 0.953 & 984 & 0.805 & 0.877 & 0.959 \\
40 & $\beta_{3}$ & active & 869 & 0.869 & 984 & 0.818 & 0.884 & 0.947 \\
40 & $\beta_{4}$ & active & 961 & 0.961 & 984 & 0.801 & 0.861 & 0.968 \\
40 & $\beta_{5}$ & inactive & 386 & 0.386 & 984 & 0.873 & 0.915 & 0.972 \\
40 & $\beta_{6}$ & inactive & 382 & 0.382 & 984 & 0.880 & 0.939 & 0.979 \\
60 & $\beta_{1}$ & active & 1000 & 1.000 & 1000 & 0.822 & 0.891 & 0.968 \\
60 & $\beta_{2}$ & active & 997 & 0.997 & 1000 & 0.823 & 0.896 & 0.971 \\
60 & $\beta_{3}$ & active & 978 & 0.978 & 1000 & 0.833 & 0.899 & 0.972 \\
60 & $\beta_{4}$ & active & 1000 & 1.000 & 1000 & 0.814 & 0.899 & 0.972 \\
60 & $\beta_{5}$ & inactive & 430 & 0.430 & 1000 & 0.891 & 0.946 & 0.986 \\
60 & $\beta_{6}$ & inactive & 454 & 0.454 & 1000 & 0.899 & 0.938 & 0.974 \\
75 & $\beta_{1}$ & active & 1000 & 1.000 & 1000 & 0.844 & 0.901 & 0.970 \\
75 & $\beta_{2}$ & active & 1000 & 1.000 & 1000 & 0.845 & 0.907 & 0.975 \\
75 & $\beta_{3}$ & active & 990 & 0.990 & 1000 & 0.833 & 0.900 & 0.953 \\
75 & $\beta_{4}$ & active & 1000 & 1.000 & 1000 & 0.831 & 0.896 & 0.975 \\
75 & $\beta_{5}$ & inactive & 463 & 0.463 & 1000 & 0.913 & 0.954 & 0.977 \\
75 & $\beta_{6}$ & inactive & 456 & 0.456 & 1000 & 0.895 & 0.945 & 0.971 \\
80 & $\beta_{1}$ & active & 1000 & 1.000 & 1000 & 0.827 & 0.905 & 0.981 \\
80 & $\beta_{2}$ & active & 1000 & 1.000 & 1000 & 0.840 & 0.902 & 0.974 \\
80 & $\beta_{3}$ & active & 993 & 0.993 & 1000 & 0.847 & 0.902 & 0.969 \\
80 & $\beta_{4}$ & active & 999 & 0.999 & 1000 & 0.848 & 0.914 & 0.978 \\
80 & $\beta_{5}$ & inactive & 466 & 0.466 & 1000 & 0.894 & 0.940 & 0.974 \\
80 & $\beta_{6}$ & inactive & 446 & 0.446 & 1000 & 0.895 & 0.937 & 0.983 \\
100 & $\beta_{1}$ & active & 1000 & 1.000 & 1000 & 0.861 & 0.923 & 0.982 \\
100 & $\beta_{2}$ & active & 1000 & 1.000 & 1000 & 0.860 & 0.929 & 0.981 \\
100 & $\beta_{3}$ & active & 1000 & 1.000 & 1000 & 0.858 & 0.933 & 0.976 \\
100 & $\beta_{4}$ & active & 1000 & 1.000 & 1000 & 0.851 & 0.913 & 0.981 \\
100 & $\beta_{5}$ & inactive & 479 & 0.479 & 1000 & 0.892 & 0.943 & 0.975 \\
100 & $\beta_{6}$ & inactive & 505 & 0.505 & 1000 & 0.897 & 0.943 & 0.971 \\
125 & $\beta_{1}$ & active & 1000 & 1.000 & 1000 & 0.860 & 0.925 & 0.983 \\
125 & $\beta_{2}$ & active & 1000 & 1.000 & 1000 & 0.844 & 0.920 & 0.976 \\
125 & $\beta_{3}$ & active & 999 & 0.999 & 1000 & 0.849 & 0.904 & 0.972 \\
125 & $\beta_{4}$ & active & 1000 & 1.000 & 1000 & 0.858 & 0.916 & 0.975 \\
125 & $\beta_{5}$ & inactive & 511 & 0.511 & 1000 & 0.895 & 0.950 & 0.974 \\
125 & $\beta_{6}$ & inactive & 513 & 0.513 & 1000 & 0.913 & 0.957 & 0.981 \\
175 & $\beta_{1}$ & active & 1000 & 1.000 & 1000 & 0.870 & 0.941 & 0.985 \\
175 & $\beta_{2}$ & active & 1000 & 1.000 & 1000 & 0.868 & 0.919 & 0.981 \\
175 & $\beta_{3}$ & active & 1000 & 1.000 & 1000 & 0.840 & 0.913 & 0.971 \\
175 & $\beta_{4}$ & active & 1000 & 1.000 & 1000 & 0.879 & 0.930 & 0.986 \\
175 & $\beta_{5}$ & inactive & 549 & 0.549 & 1000 & 0.896 & 0.943 & 0.971 \\
175 & $\beta_{6}$ & inactive & 565 & 0.565 & 1000 & 0.894 & 0.939 & 0.964 \\
200 & $\beta_{1}$ & active & 1000 & 1.000 & 1000 & 0.890 & 0.936 & 0.987 \\
200 & $\beta_{2}$ & active & 1000 & 1.000 & 1000 & 0.883 & 0.933 & 0.975 \\
200 & $\beta_{3}$ & active & 1000 & 1.000 & 1000 & 0.878 & 0.931 & 0.980 \\
200 & $\beta_{4}$ & active & 1000 & 1.000 & 1000 & 0.879 & 0.942 & 0.991 \\
200 & $\beta_{5}$ & inactive & 541 & 0.541 & 1000 & 0.899 & 0.952 & 0.970 \\
200 & $\beta_{6}$ & inactive & 554 & 0.554 & 1000 & 0.891 & 0.942 & 0.976 \\
250 & $\beta_{1}$ & active & 1000 & 1.000 & 1000 & 0.887 & 0.944 & 0.985 \\
250 & $\beta_{2}$ & active & 1000 & 1.000 & 1000 & 0.875 & 0.932 & 0.987 \\
250 & $\beta_{3}$ & active & 1000 & 1.000 & 1000 & 0.858 & 0.911 & 0.969 \\
250 & $\beta_{4}$ & active & 1000 & 1.000 & 1000 & 0.887 & 0.938 & 0.985 \\
250 & $\beta_{5}$ & inactive & 557 & 0.557 & 1000 & 0.909 & 0.954 & 0.973 \\
250 & $\beta_{6}$ & inactive & 564 & 0.564 & 1000 & 0.889 & 0.951 & 0.968 \\
\end{longtable}
\endgroup

%% file: generated_supplement_tables_v8/table_coeff_coverage_aic_adapttrue.tex
\begingroup
\scriptsize
\setlength{\tabcolsep}{3pt}
\begin{longtable}{rrlrrrrrr}
\caption{Coefficient-specific empirical coverage for Reference setting~R2. The table reports Cox Lasso selection counts and frequencies from the original simulated data sets separately from interval coverage. $N_{\mathrm{CI}}$ is the number of Monte Carlo replications for which all three debiased intervals were available.}\label{tab:supp-coeff-coverage-aic}\\
\toprule
$n$ & Coef. & Status & $N_{\mathrm{sel}}$ & Sel. freq. & $N_{\mathrm{CI}}$ & Deb. Wald & WB basic & Ef. basic \\
\midrule
\endfirsthead
\caption[]{Coefficient-specific empirical coverage for Reference setting~R2. (continued)}\\
\toprule
$n$ & Coef. & Status & $N_{\mathrm{sel}}$ & Sel. freq. & $N_{\mathrm{CI}}$ & Deb. Wald & WB basic & Ef. basic \\
\midrule
\endhead
\midrule
\multicolumn{9}{r}{Continued on next page}\\
\endfoot
\bottomrule
\endlastfoot
30 & $\beta_{1}$ & active & 957 & 0.957 & 995 & 0.733 & 0.822 & 0.864 \\
30 & $\beta_{2}$ & active & 929 & 0.929 & 995 & 0.729 & 0.820 & 0.852 \\
30 & $\beta_{3}$ & active & 799 & 0.799 & 995 & 0.741 & 0.820 & 0.849 \\
30 & $\beta_{4}$ & active & 960 & 0.960 & 995 & 0.746 & 0.833 & 0.866 \\
30 & $\beta_{5}$ & inactive & 361 & 0.361 & 995 & 0.772 & 0.837 & 0.844 \\
30 & $\beta_{6}$ & inactive & 364 & 0.364 & 995 & 0.767 & 0.825 & 0.846 \\
40 & $\beta_{1}$ & active & 989 & 0.989 & 1000 & 0.762 & 0.851 & 0.885 \\
40 & $\beta_{2}$ & active & 975 & 0.975 & 1000 & 0.798 & 0.861 & 0.896 \\
40 & $\beta_{3}$ & active & 904 & 0.904 & 1000 & 0.807 & 0.873 & 0.906 \\
40 & $\beta_{4}$ & active & 992 & 0.992 & 1000 & 0.777 & 0.853 & 0.895 \\
40 & $\beta_{5}$ & inactive & 311 & 0.311 & 1000 & 0.808 & 0.877 & 0.889 \\
40 & $\beta_{6}$ & inactive & 329 & 0.329 & 1000 & 0.821 & 0.885 & 0.904 \\
60 & $\beta_{1}$ & active & 1000 & 1.000 & 1000 & 0.793 & 0.881 & 0.915 \\
60 & $\beta_{2}$ & active & 993 & 0.993 & 1000 & 0.825 & 0.886 & 0.925 \\
60 & $\beta_{3}$ & active & 966 & 0.966 & 1000 & 0.837 & 0.887 & 0.915 \\
60 & $\beta_{4}$ & active & 1000 & 1.000 & 1000 & 0.817 & 0.889 & 0.920 \\
60 & $\beta_{5}$ & inactive & 272 & 0.272 & 1000 & 0.850 & 0.914 & 0.919 \\
60 & $\beta_{6}$ & inactive & 278 & 0.278 & 1000 & 0.845 & 0.907 & 0.924 \\
75 & $\beta_{1}$ & active & 1000 & 1.000 & 1000 & 0.868 & 0.928 & 0.952 \\
75 & $\beta_{2}$ & active & 999 & 0.999 & 1000 & 0.856 & 0.911 & 0.937 \\
75 & $\beta_{3}$ & active & 982 & 0.982 & 1000 & 0.836 & 0.894 & 0.914 \\
75 & $\beta_{4}$ & active & 1000 & 1.000 & 1000 & 0.826 & 0.900 & 0.938 \\
75 & $\beta_{5}$ & inactive & 246 & 0.246 & 1000 & 0.850 & 0.921 & 0.931 \\
75 & $\beta_{6}$ & inactive & 248 & 0.248 & 1000 & 0.862 & 0.918 & 0.926 \\
80 & $\beta_{1}$ & active & 1000 & 1.000 & 1000 & 0.836 & 0.896 & 0.933 \\
80 & $\beta_{2}$ & active & 1000 & 1.000 & 1000 & 0.840 & 0.899 & 0.937 \\
80 & $\beta_{3}$ & active & 988 & 0.988 & 1000 & 0.842 & 0.912 & 0.941 \\
80 & $\beta_{4}$ & active & 1000 & 1.000 & 1000 & 0.855 & 0.917 & 0.941 \\
80 & $\beta_{5}$ & inactive & 247 & 0.247 & 1000 & 0.869 & 0.924 & 0.939 \\
80 & $\beta_{6}$ & inactive & 245 & 0.245 & 1000 & 0.862 & 0.920 & 0.930 \\
100 & $\beta_{1}$ & active & 1000 & 1.000 & 1000 & 0.866 & 0.923 & 0.952 \\
100 & $\beta_{2}$ & active & 1000 & 1.000 & 1000 & 0.864 & 0.910 & 0.934 \\
100 & $\beta_{3}$ & active & 997 & 0.997 & 1000 & 0.870 & 0.929 & 0.947 \\
100 & $\beta_{4}$ & active & 1000 & 1.000 & 1000 & 0.845 & 0.904 & 0.946 \\
100 & $\beta_{5}$ & inactive & 224 & 0.224 & 1000 & 0.877 & 0.925 & 0.943 \\
100 & $\beta_{6}$ & inactive & 198 & 0.198 & 1000 & 0.897 & 0.934 & 0.940 \\
125 & $\beta_{1}$ & active & 1000 & 1.000 & 1000 & 0.861 & 0.922 & 0.951 \\
125 & $\beta_{2}$ & active & 1000 & 1.000 & 1000 & 0.870 & 0.934 & 0.959 \\
125 & $\beta_{3}$ & active & 999 & 0.999 & 1000 & 0.863 & 0.931 & 0.958 \\
125 & $\beta_{4}$ & active & 1000 & 1.000 & 1000 & 0.857 & 0.911 & 0.945 \\
125 & $\beta_{5}$ & inactive & 199 & 0.199 & 1000 & 0.870 & 0.929 & 0.937 \\
125 & $\beta_{6}$ & inactive & 217 & 0.217 & 1000 & 0.887 & 0.938 & 0.948 \\
175 & $\beta_{1}$ & active & 1000 & 1.000 & 1000 & 0.868 & 0.935 & 0.954 \\
175 & $\beta_{2}$ & active & 1000 & 1.000 & 1000 & 0.868 & 0.928 & 0.961 \\
175 & $\beta_{3}$ & active & 1000 & 1.000 & 1000 & 0.874 & 0.928 & 0.960 \\
175 & $\beta_{4}$ & active & 1000 & 1.000 & 1000 & 0.870 & 0.933 & 0.956 \\
175 & $\beta_{5}$ & inactive & 185 & 0.185 & 1000 & 0.892 & 0.944 & 0.956 \\
175 & $\beta_{6}$ & inactive & 215 & 0.215 & 1000 & 0.868 & 0.929 & 0.944 \\
200 & $\beta_{1}$ & active & 1000 & 1.000 & 1000 & 0.878 & 0.935 & 0.967 \\
200 & $\beta_{2}$ & active & 1000 & 1.000 & 1000 & 0.877 & 0.936 & 0.964 \\
200 & $\beta_{3}$ & active & 1000 & 1.000 & 1000 & 0.880 & 0.942 & 0.966 \\
200 & $\beta_{4}$ & active & 1000 & 1.000 & 1000 & 0.880 & 0.938 & 0.966 \\
200 & $\beta_{5}$ & inactive & 201 & 0.201 & 1000 & 0.872 & 0.930 & 0.940 \\
200 & $\beta_{6}$ & inactive & 213 & 0.213 & 1000 & 0.878 & 0.931 & 0.950 \\
250 & $\beta_{1}$ & active & 1000 & 1.000 & 1000 & 0.884 & 0.929 & 0.973 \\
250 & $\beta_{2}$ & active & 1000 & 1.000 & 1000 & 0.884 & 0.943 & 0.967 \\
250 & $\beta_{3}$ & active & 1000 & 1.000 & 1000 & 0.875 & 0.935 & 0.965 \\
250 & $\beta_{4}$ & active & 1000 & 1.000 & 1000 & 0.884 & 0.932 & 0.963 \\
250 & $\beta_{5}$ & inactive & 167 & 0.167 & 1000 & 0.908 & 0.948 & 0.957 \\
250 & $\beta_{6}$ & inactive & 188 & 0.188 & 1000 & 0.883 & 0.938 & 0.944 \\
\end{longtable}
\endgroup

%% file: generated_supplement_tables_v9/table_selection_frequency_tuning_grid_n80.tex
\begin{table}[p]
\centering
\scriptsize
\setlength{\tabcolsep}{4pt}
\caption{Cox Lasso selection counts and frequencies for a simulation setting
with \(n=80\), \(p=10\), the realistic coefficient pattern, a Weibull baseline
hazard, 10\% target censoring, and independent covariates. Each cell gives
\(N_{\mathrm{sel},j}\) followed by the empirical selection frequency in
parentheses. All eight combinations of tuning rule and penalty specification
are available for this setting, based on \(M=1000\) Monte Carlo replications.}
\label{tab:supp-selection-tuning-grid}
\begin{tabular}{llcccc}
\toprule
Coef. & Status & $\lambda_{\min}$ & $\lambda_{1se}$ & AIC & BIC \\
\midrule
\multicolumn{6}{l}{\textit{Panel A: Standard Cox Lasso (non-adaptive)}} \\
\addlinespace
$\beta_{1}$ & active & 1000 (1.000) & 997 (0.997) & 1000 (1.000) & 1000 (1.000) \\
$\beta_{2}$ & active & 1000 (1.000) & 983 (0.983) & 1000 (1.000) & 1000 (1.000) \\
$\beta_{3}$ & active & 993 (0.993) & 909 (0.909) & 998 (0.998) & 983 (0.983) \\
$\beta_{4}$ & active & 999 (0.999) & 997 (0.997) & 1000 (1.000) & 1000 (1.000) \\
$\beta_{5}$ & inactive & 466 (0.466) & 101 (0.101) & 573 (0.573) & 330 (0.330) \\
$\beta_{6}$ & inactive & 446 (0.446) & 120 (0.120) & 559 (0.559) & 326 (0.326) \\
\midrule
\multicolumn{6}{l}{\textit{Panel B: Adaptive Cox Lasso}} \\
\addlinespace
$\beta_{1}$ & active & 1000 (1.000) & 987 (0.987) & 1000 (1.000) & 1000 (1.000) \\
$\beta_{2}$ & active & 1000 (1.000) & 953 (0.953) & 1000 (1.000) & 997 (0.997) \\
$\beta_{3}$ & active & 968 (0.968) & 687 (0.687) & 988 (0.988) & 959 (0.959) \\
$\beta_{4}$ & active & 1000 (1.000) & 980 (0.980) & 1000 (1.000) & 999 (0.999) \\
$\beta_{5}$ & inactive & 154 (0.154) & 5 (0.005) & 247 (0.247) & 94 (0.094) \\
$\beta_{6}$ & inactive & 140 (0.140) & 9 (0.009) & 245 (0.245) & 96 (0.096) \\
\bottomrule
\end{tabular}
\end{table}

%% file: generated_supplement_tables_v10/table_tuning_penalty_mean_active_coverage_n80.tex
\begin{table}[p]
\centering
\small
\setlength{\tabcolsep}{4pt}
\caption{Mean empirical coverage for active coefficients across all eight combinations of tuning rule and penalty specification under the simulation setting used in Table~\ref{tab:supp-selection-tuning-grid}: $n=80$, $p=10$, the \textit{realistic} coefficient pattern, Weibull baseline hazard, 10\% target censoring, and independent covariates. Coverage is averaged over $\beta_1,\ldots,\beta_4$. `Sel. freq.' is the mean Cox Lasso selection frequency of these four active coefficients in the original simulated data sets. $N_{\mathrm{CI}}$ is the common number of Monte Carlo replications in which the three debiased intervals were available.}
\label{tab:supp-tuning-coverage}
\begin{tabular}{llrrrrr}
\toprule
Penalty & Tuning & Sel. freq. & $N_{\mathrm{CI}}$ & Deb. Wald & WB basic & Ef. basic \\
\midrule
Standard & $\lambda_{\min}$ & 0.998 & 1000 & 0.840 & 0.906 & 0.976 \\
Standard & $\lambda_{1se}$ & 0.972 & 999 & 0.789 & 0.866 & 0.986 \\
Standard & AIC & 1.000 & 1000 & 0.838 & 0.901 & 0.932 \\
Standard & BIC & 0.996 & 1000 & 0.835 & 0.906 & 0.968 \\
\midrule
Adaptive & $\lambda_{\min}$ & 0.992 & 1000 & 0.838 & 0.905 & 0.974 \\
Adaptive & $\lambda_{1se}$ & 0.902 & 996 & 0.798 & 0.876 & 0.985 \\
Adaptive & AIC & 0.997 & 1000 & 0.843 & 0.906 & 0.938 \\
Adaptive & BIC & 0.989 & 1000 & 0.840 & 0.908 & 0.968 \\
\bottomrule
\end{tabular}
\end{table}

%% file: generated_supplement_tables_v10/table_tuning_penalty_mean_active_width_n80.tex
\begin{table}[p]
\centering
\small
\setlength{\tabcolsep}{5pt}
\caption{Mean 90\% confidence-interval width for active coefficients across the same eight combinations of tuning rule and penalty specification and under the same simulation setting as Tables~\ref{tab:supp-selection-tuning-grid} and~\ref{tab:supp-tuning-coverage}. Widths are averaged over $\beta_1,\ldots,\beta_4$.}
\label{tab:supp-tuning-width}
\begin{tabular}{llrrr}
\toprule
Penalty & Tuning & Deb. Wald & WB basic & Ef. basic \\
\midrule
Standard & $\lambda_{\min}$ & 0.495 & 0.588 & 0.814 \\
Standard & $\lambda_{1se}$ & 0.497 & 0.591 & 1.142 \\
Standard & AIC & 0.495 & 0.588 & 0.670 \\
Standard & BIC & 0.495 & 0.588 & 0.809 \\
\midrule
Adaptive & $\lambda_{\min}$ & 0.494 & 0.587 & 0.817 \\
Adaptive & $\lambda_{1se}$ & 0.498 & 0.592 & 1.150 \\
Adaptive & AIC & 0.493 & 0.586 & 0.668 \\
Adaptive & BIC & 0.493 & 0.586 & 0.806 \\
\bottomrule
\end{tabular}
\end{table}

%% file: generated_supplement_tables_v8/table_coeff_width_min_adaptfalse.tex
\begingroup
\scriptsize
\setlength{\tabcolsep}{4pt}
\begin{longtable}{rrlrrr}
\caption{Coefficient-specific mean widths of the 90\% confidence intervals
for Reference setting~R1. The table reports the debiased Wald interval and
the basic intervals based on the wild bootstrap and Efron's bootstrap.}
\label{tab:supp-coeff-width-min}\\
\toprule
$n$ & Coef. & Status & Deb. Wald & WB basic & Ef. basic \\
\midrule
\endfirsthead
\caption[]{Coefficient-specific mean interval widths for Reference setting~R1. (continued)}\\
\toprule
$n$ & Coef. & Status & Deb. Wald & WB basic & Ef. basic \\
\midrule
\endhead
\midrule
\multicolumn{6}{r}{Continued on next page}\\
\endfoot
\bottomrule
\endlastfoot
30 & $\beta_{1}$ & active & 1.046 & 1.243 & 1.963 \\
30 & $\beta_{2}$ & active & 0.997 & 1.184 & 1.821 \\
30 & $\beta_{3}$ & active & 0.904 & 1.073 & 1.545 \\
30 & $\beta_{4}$ & active & 1.063 & 1.264 & 2.007 \\
30 & $\beta_{5}$ & inactive & 0.800 & 0.950 & 1.154 \\
30 & $\beta_{6}$ & inactive & 0.803 & 0.953 & 1.164 \\
40 & $\beta_{1}$ & active & 0.805 & 0.956 & 1.419 \\
40 & $\beta_{2}$ & active & 0.783 & 0.931 & 1.347 \\
40 & $\beta_{3}$ & active & 0.704 & 0.836 & 1.142 \\
40 & $\beta_{4}$ & active & 0.807 & 0.959 & 1.423 \\
40 & $\beta_{5}$ & inactive & 0.646 & 0.766 & 0.905 \\
40 & $\beta_{6}$ & inactive & 0.639 & 0.761 & 0.888 \\
60 & $\beta_{1}$ & active & 0.605 & 0.719 & 1.025 \\
60 & $\beta_{2}$ & active & 0.584 & 0.694 & 0.971 \\
60 & $\beta_{3}$ & active & 0.541 & 0.642 & 0.860 \\
60 & $\beta_{4}$ & active & 0.614 & 0.729 & 1.043 \\
60 & $\beta_{5}$ & inactive & 0.496 & 0.590 & 0.672 \\
60 & $\beta_{6}$ & inactive & 0.497 & 0.591 & 0.676 \\
75 & $\beta_{1}$ & active & 0.531 & 0.631 & 0.890 \\
75 & $\beta_{2}$ & active & 0.511 & 0.608 & 0.842 \\
75 & $\beta_{3}$ & active & 0.470 & 0.559 & 0.733 \\
75 & $\beta_{4}$ & active & 0.535 & 0.635 & 0.900 \\
75 & $\beta_{5}$ & inactive & 0.433 & 0.515 & 0.579 \\
75 & $\beta_{6}$ & inactive & 0.434 & 0.515 & 0.580 \\
80 & $\beta_{1}$ & active & 0.514 & 0.610 & 0.864 \\
80 & $\beta_{2}$ & active & 0.492 & 0.585 & 0.806 \\
80 & $\beta_{3}$ & active & 0.457 & 0.543 & 0.718 \\
80 & $\beta_{4}$ & active & 0.515 & 0.612 & 0.867 \\
80 & $\beta_{5}$ & inactive & 0.419 & 0.498 & 0.563 \\
80 & $\beta_{6}$ & inactive & 0.420 & 0.500 & 0.561 \\
100 & $\beta_{1}$ & active & 0.454 & 0.539 & 0.758 \\
100 & $\beta_{2}$ & active & 0.431 & 0.513 & 0.702 \\
100 & $\beta_{3}$ & active & 0.401 & 0.477 & 0.623 \\
100 & $\beta_{4}$ & active & 0.454 & 0.540 & 0.760 \\
100 & $\beta_{5}$ & inactive & 0.368 & 0.438 & 0.490 \\
100 & $\beta_{6}$ & inactive & 0.368 & 0.438 & 0.489 \\
125 & $\beta_{1}$ & active & 0.397 & 0.472 & 0.657 \\
125 & $\beta_{2}$ & active & 0.381 & 0.453 & 0.619 \\
125 & $\beta_{3}$ & active & 0.355 & 0.423 & 0.548 \\
125 & $\beta_{4}$ & active & 0.397 & 0.472 & 0.655 \\
125 & $\beta_{5}$ & inactive & 0.325 & 0.387 & 0.430 \\
125 & $\beta_{6}$ & inactive & 0.325 & 0.386 & 0.431 \\
175 & $\beta_{1}$ & active & 0.330 & 0.393 & 0.540 \\
175 & $\beta_{2}$ & active & 0.318 & 0.378 & 0.511 \\
175 & $\beta_{3}$ & active & 0.295 & 0.350 & 0.450 \\
175 & $\beta_{4}$ & active & 0.330 & 0.393 & 0.539 \\
175 & $\beta_{5}$ & inactive & 0.271 & 0.322 & 0.355 \\
175 & $\beta_{6}$ & inactive & 0.271 & 0.322 & 0.356 \\
200 & $\beta_{1}$ & active & 0.306 & 0.364 & 0.504 \\
200 & $\beta_{2}$ & active & 0.296 & 0.351 & 0.476 \\
200 & $\beta_{3}$ & active & 0.276 & 0.328 & 0.424 \\
200 & $\beta_{4}$ & active & 0.307 & 0.366 & 0.506 \\
200 & $\beta_{5}$ & inactive & 0.253 & 0.300 & 0.332 \\
200 & $\beta_{6}$ & inactive & 0.253 & 0.301 & 0.332 \\
250 & $\beta_{1}$ & active & 0.272 & 0.323 & 0.444 \\
250 & $\beta_{2}$ & active & 0.261 & 0.311 & 0.419 \\
250 & $\beta_{3}$ & active & 0.245 & 0.291 & 0.375 \\
250 & $\beta_{4}$ & active & 0.272 & 0.323 & 0.444 \\
250 & $\beta_{5}$ & inactive & 0.225 & 0.267 & 0.292 \\
250 & $\beta_{6}$ & inactive & 0.225 & 0.268 & 0.294 \\
\end{longtable}
\endgroup

%% file: generated_supplement_tables_v8/table_coeff_width_aic_adapttrue.tex
\begingroup
\scriptsize
\setlength{\tabcolsep}{4pt}
\begin{longtable}{rrlrrr}
\caption{Coefficient-specific mean widths of the 90\% confidence intervals
for Reference setting~R2. The table reports the debiased Wald interval and
the basic intervals based on the wild bootstrap and Efron's bootstrap.}
\label{tab:supp-coeff-width-aic}\\
\toprule
$n$ & Coef. & Status & Deb. Wald & WB basic & Ef. basic \\
\midrule
\endfirsthead
\caption[]{Coefficient-specific mean interval widths for Reference setting~R2. (continued)}\\
\toprule
$n$ & Coef. & Status & Deb. Wald & WB basic & Ef. basic \\
\midrule
\endhead
\midrule
\multicolumn{6}{r}{Continued on next page}\\
\endfoot
\bottomrule
\endlastfoot
30 & $\beta_{1}$ & active & 1.086 & 1.291 & 1.438 \\
30 & $\beta_{2}$ & active & 1.052 & 1.250 & 1.389 \\
30 & $\beta_{3}$ & active & 0.944 & 1.121 & 1.222 \\
30 & $\beta_{4}$ & active & 1.105 & 1.312 & 1.469 \\
30 & $\beta_{5}$ & inactive & 0.844 & 1.002 & 1.054 \\
30 & $\beta_{6}$ & inactive & 0.851 & 1.011 & 1.065 \\
40 & $\beta_{1}$ & active & 0.819 & 0.975 & 1.102 \\
40 & $\beta_{2}$ & active & 0.782 & 0.929 & 1.043 \\
40 & $\beta_{3}$ & active & 0.713 & 0.847 & 0.932 \\
40 & $\beta_{4}$ & active & 0.821 & 0.977 & 1.105 \\
40 & $\beta_{5}$ & inactive & 0.656 & 0.780 & 0.821 \\
40 & $\beta_{6}$ & inactive & 0.657 & 0.781 & 0.821 \\
60 & $\beta_{1}$ & active & 0.613 & 0.729 & 0.827 \\
60 & $\beta_{2}$ & active & 0.587 & 0.699 & 0.785 \\
60 & $\beta_{3}$ & active & 0.543 & 0.645 & 0.711 \\
60 & $\beta_{4}$ & active & 0.613 & 0.728 & 0.829 \\
60 & $\beta_{5}$ & inactive & 0.502 & 0.596 & 0.625 \\
60 & $\beta_{6}$ & inactive & 0.503 & 0.599 & 0.628 \\
75 & $\beta_{1}$ & active & 0.528 & 0.629 & 0.718 \\
75 & $\beta_{2}$ & active & 0.511 & 0.606 & 0.688 \\
75 & $\beta_{3}$ & active & 0.473 & 0.563 & 0.624 \\
75 & $\beta_{4}$ & active & 0.537 & 0.638 & 0.733 \\
75 & $\beta_{5}$ & inactive & 0.434 & 0.515 & 0.540 \\
75 & $\beta_{6}$ & inactive & 0.434 & 0.516 & 0.540 \\
80 & $\beta_{1}$ & active & 0.511 & 0.608 & 0.699 \\
80 & $\beta_{2}$ & active & 0.492 & 0.583 & 0.665 \\
80 & $\beta_{3}$ & active & 0.456 & 0.542 & 0.604 \\
80 & $\beta_{4}$ & active & 0.515 & 0.611 & 0.704 \\
80 & $\beta_{5}$ & inactive & 0.420 & 0.499 & 0.523 \\
80 & $\beta_{6}$ & inactive & 0.419 & 0.497 & 0.522 \\
100 & $\beta_{1}$ & active & 0.453 & 0.537 & 0.627 \\
100 & $\beta_{2}$ & active & 0.433 & 0.514 & 0.594 \\
100 & $\beta_{3}$ & active & 0.401 & 0.477 & 0.534 \\
100 & $\beta_{4}$ & active & 0.453 & 0.538 & 0.628 \\
100 & $\beta_{5}$ & inactive & 0.370 & 0.438 & 0.461 \\
100 & $\beta_{6}$ & inactive & 0.369 & 0.438 & 0.459 \\
125 & $\beta_{1}$ & active & 0.396 & 0.472 & 0.551 \\
125 & $\beta_{2}$ & active & 0.382 & 0.454 & 0.525 \\
125 & $\beta_{3}$ & active & 0.354 & 0.420 & 0.473 \\
125 & $\beta_{4}$ & active & 0.395 & 0.470 & 0.549 \\
125 & $\beta_{5}$ & inactive & 0.325 & 0.387 & 0.406 \\
125 & $\beta_{6}$ & inactive & 0.326 & 0.387 & 0.406 \\
175 & $\beta_{1}$ & active & 0.330 & 0.391 & 0.460 \\
175 & $\beta_{2}$ & active & 0.317 & 0.376 & 0.439 \\
175 & $\beta_{3}$ & active & 0.295 & 0.351 & 0.396 \\
175 & $\beta_{4}$ & active & 0.328 & 0.390 & 0.458 \\
175 & $\beta_{5}$ & inactive & 0.271 & 0.321 & 0.337 \\
175 & $\beta_{6}$ & inactive & 0.271 & 0.323 & 0.337 \\
200 & $\beta_{1}$ & active & 0.306 & 0.363 & 0.433 \\
200 & $\beta_{2}$ & active & 0.296 & 0.352 & 0.415 \\
200 & $\beta_{3}$ & active & 0.275 & 0.327 & 0.374 \\
200 & $\beta_{4}$ & active & 0.306 & 0.364 & 0.433 \\
200 & $\beta_{5}$ & inactive & 0.253 & 0.301 & 0.317 \\
200 & $\beta_{6}$ & inactive & 0.253 & 0.300 & 0.316 \\
250 & $\beta_{1}$ & active & 0.272 & 0.322 & 0.382 \\
250 & $\beta_{2}$ & active & 0.262 & 0.312 & 0.366 \\
250 & $\beta_{3}$ & active & 0.245 & 0.291 & 0.332 \\
250 & $\beta_{4}$ & active & 0.272 & 0.323 & 0.384 \\
250 & $\beta_{5}$ & inactive & 0.225 & 0.268 & 0.281 \\
250 & $\beta_{6}$ & inactive & 0.225 & 0.267 & 0.281 \\
\end{longtable}
\endgroup